\documentclass[final]{lmcs}
\usepackage[utf8]{inputenc}
\pdfoutput=1

\usepackage{lastpage}
\lmcsdoi{22}{1}{27}
\lmcsheading{}{\pageref{LastPage}}{}{}%
{Jan.~30,~2023}{Mar.~17,~2026}{}

\usepackage{breakcites}
\usepackage{preamble}
\usepackage{macros}
\usepackage{jon-tikz}
\usepackage{categories}
\usepackage{hyperref}

\begin{document}

\title{Normalization for Multimodal Type Theory}
\author{Daniel Gratzer\lmcsorcid{0000-0003-1944-0789}}
\address{Aarhus University}
\email{gratzer@cs.au.dk}

\begin{abstract}
  We prove normalization for \MTT{}, a general multimodal dependent type theory capable of
  expressing modal type theories for guarded recursion, internalized parametricity, and various
  other prototypical modal situations. We prove that deciding type checking and conversion in \MTT{}
  can be reduced to deciding the equality of modalities in the underlying modal situation,
  immediately yielding a type checking algorithm for all instantiations of \MTT{} in the literature.
  This proof uses a generalization of \emph{synthetic Tait computability}---an abstract approach to
  gluing proofs---to account for modalities. This extension is based on \MTT{} itself, so that this
  proof also constitutes a significant case study of \MTT{}.
\end{abstract}

\maketitle

\section{Introduction}
\label{sec:introduction}

If type theory is classically the study of objects invariant under change of context, modal type
theory is the study of adding non-invariant connectives---\emph{modalities}---to type theory. Given
that many natural features of particular models of type theory are not invariant under substitution,
modal type theories have sparked considerable interest. By nature, however, modal type theories must
thread the needle of presenting modalities in such a way that the classical substitution theorems of
type theory still hold.

Typically, modal type theories require modifications to the apparatus of contexts and
substitutions. Unfortunately, these tweaks are often more art than science, with expert attention
required even to make the most trivial modification to the modal structure of a type theory. In
order to address this complexity, \emph{general} modal type theories have been
introduced~\cite{licata:2017,gratzer:2020}. These theories can be instantiated by a description of
a modal situation to produce a system enjoying the theorems usually proved by experts.

\subsection{Multimodal type theory}

We focus on one such general modal type theory: \MTT{} \cite{gratzer:2020}. \MTT{} can be
instantiated with an arbitrary collection of modalities and transformations between them to yield a
highly usable syntax. The modalities in \MTT{} behave like (weak) dependent right adjoints
(DRAs)~\cite{birkedal:2020} so that \MTT{} can be used to internalize nearly any right
adjoint. This flexibility allows \MTT{} to encode calculi for guarded recursion, internalized
parametricity, and other handcrafted calculi.

More precisely, \MTT{} can be instantiated by a \emph{mode theory}, a strict 2-category describing
modes, modalities, and natural transformations between these modalities. This 2-categorical
structure is then reflected into the structure of substitutions in \MTT{}, ensuring that \eg{}, a
transformation between two modalities $\mu$ and $\nu$ gives rise to a function
$\Modify{A} \to \Modify[\nu]{A}$.

While this flexibility allows \MTT{} to accommodate many interesting calculi, it becomes
proportionally more challenging to prove metatheoretic results about \MTT{}. In particular, the rich
substitution structure inherited from the mode theory can introduce subtle equations between
terms. The proof that the crisp induction principles can be reconstructed in \MTT{}~\cite[Theorem
10.4]{gratzer:mtt-journal:2021}, for instance, exemplifies this and hinges on many such
calculations. In fact, the metatheoretic results established by Gratzer et al.~\cite{gratzer:2020}
(soundness and canonicity) are results on closed terms in \MTT{}, allowing their proofs to avoid the
majority of the substitution apparatus.

Crucially, it remained open whether \MTT{} admitted a normalization algorithm and, consequently,
whether type checking was decidable. Even in the presence of a normalization algorithm \MTT{} cannot
admit an unconditional type checking algorithm: it is not only necessary to have a decision
procedure for terms in the language, but also for modalities and 2-cells as both appear in terms for
\MTT{}.

In this paper we show the best possible result holds: \MTT{} admits an unconditional normalization
algorithm and conversion of normal forms is decidable if conversion is decidable in the mode
theory.\footnote{The converse is almost, but not quite, true. Decidability of conversion for normal
  forms implies that the 1- and 2-cells of the mode theory have decidable equality, as these appear
  in normal forms.} As corollaries, we show that type constructors in \MTT{} are always injective
and that type checking is decidable when the mode theory is decidable.%
\footnote{This requirement is potentially nontrivial \eg{}, the word problem for groups is known to
  be undecidable and is subsumed by the problem for 2-categories.}

\subsection{Normalization-by-evaluation}
\label{sec:introduction:normalization}
A normalization algorithm must begin by defining \emph{normal forms}. Their precise formulation
depends on the situation but they always satisfy two crucial properties. First, the equality
of normal forms $u = v$ is clearly decidable---often no more than structural equality---and there is
a function $\DecodeNf{u}$ decoding a normal form to a term of the same type.

Relative to a notion of normal form, a normalization algorithm sends a term $\Gamma \vdash M : A$ to
a normal form $\Normalize{M}{A}$ such that $\prn{\Normalize{-}{A}, \DecodeNf{-}}$ lifts to an
isomorphism between equivalence classes of terms of $A$ and normal forms~\cite{abel:2013}. Typically
one breaks the condition that $\prn{\Normalize{-}{A},\DecodeNf{-}}$ forms an isomorphism into three
conditions:
\begin{enumerate}
\item \emph{Completeness}: if $\Gamma \vdash M = N : A$ then
  $\Normalize{M}{A} = \Normalize{N}{A}$.
\item \emph{Soundness}: $\Gamma \vdash \DecodeNf{\Normalize{M}{A}} = M : A$.
\item \emph{Idempotence}: $u = \Normalize{\DecodeNf{u}}{A}$.
\end{enumerate}

\begin{rem}
  We warn the reader that this terminology is not entirely standard. Various sources use the
  opposite conventions of soundness and completeness~\cite{altenkirch:2016,altenkirch:2017}. Such
  sources often refer to the final condition as \emph{stability}.
\end{rem}

Proving normalization is an involved affair. Traditionally, one begins by fixing a strongly
normalizing confluent rewriting system presenting the equational theory of the type theory. The
normal forms are then exactly the terms of the theory which cannot be further reduced. This approach
does not scale, however, to type theories with \emph{type-directed} equations such as the unicity
principles of dependent sums and the unit type. These equations defy attempts to present them
in a rewriting system and require type-directed algorithms.

The preeminent type-directed technique for normalization is \emph{normalization-by-evaluation
  (NbE)}~\cite{abel:2013}. Proving that an NbE algorithm works, however, is an extremely intricate
affair involving a variety of complex constructions. After the algorithm is defined, the proof of
correctness typically proceeds by establishing properties (1)-(3) in order. Each property, moreover,
requires a separate argument. Completeness is established through a PER model, soundness through a
cross-language logical relation, and idempotence through a final inductive argument. The first two
properties in particular are time-consuming to verify; recent work by Gratzer et
al.~\cite{gratzer:2019} extended NbE to a type theory with an idempotent comonad but even in this
minimal case the correctness proof occupied a 90 page technical
report~\cite{gratzer:tech-report:2019}.

These difficulties are not unique to modal type theories, and a long line of research focuses on
taming the complexity of NbE through
\emph{gluing}~\cite{altenkirch:1995,streicher:1998,fiore:2002,altenkirch:2016,coquand:2019,sterling:phd}.
This line of work recasts normalization algorithms as the construction of models of type theory in
categories defined by Artin gluing.

\subsection{Normalization-by-gluing}
\label{sec:introduction:normalization-by-gluing}

Stepping back from type theory and normalization, fix a functor $\Mor[F]{\CC}{\DD}$ between a pair
of categories. The \emph{gluing} of $F$ (written $\GL{F}$) is a category whose objects triples
$\prn{C : \CC, D : \DD, \Mor[f]{D}{F\prn{C}}}$. Morphisms in this category are given by pairs of
morphisms $\prn{x_0,x_1}$ fitting into a commuting square, \eg{}:
\[
  \DiagramSquare{
    width = 5cm,
    nw = D_0,
    ne = D_1,
    sw = F\prn{C_0},
    se = F\prn{C_1},
    north = x_1,
    south = F\prn{x_0},
    west = f_0,
    east = f_1,
  }
\]
We note that there are evident projection functors $\Mor[\pi_0]{\GL{F}}{\CC}$ and
$\Mor[\pi_1]{\GL{F}}{\DD}$.

We will view $\GL{F}$ as a category of proof-relevant predicates on $\CC$. To illustrate this,
consider $\EE = \GL{\Gamma}$ where $\Mor[\Gamma = \Hom{\ObjTerm{}}{-}]{\CC}{\SET}$ is the global
sections map on a cartesian closed category $\CC$ sending each object to the set of its global
points. Objects in $\EE$ then correspond to an object $C : \CC$ equipped with a map of sets
$\Mor[\pi]{X}{\Hom{\ObjTerm{}}{C}}$. Shifting perspective, we can view $\pi$ as a (proof-relevant)
predicate on the global points of $C$ by setting $\Phi\prn{c} = \pi^{-1}\prn{c}$.

Remarkably, $\EE$ inherits much of the structure of $\CC$ so that $\EE$ is also a Cartesian closed
category and $\pi_0$ preserves finite products and exponentials. This is a recurrent pattern with
Artin gluing; if $\Mor[F]{\CC}{\DD}$ is a nice functor between categories closed under (co)limits,
exponentials, \etc{}, then $\GL{F}$ will be closed under the same operations in such a way that
$\pi_0$ preserves them. In fact, unfolding the construction of \eg{} binary products and
exponentials in $\EE$ yields the definition familiar from logical relations.

\begin{exa}
  Viewing objects of $\EE$ as proof-relevant predicates as described above, the exponential
  $\prn{C,\Phi}^{\prn{D,\Psi}}$ is given by the following pair $\prn{C^D,\Xi}$ where $\Xi$ is
  defined as follows (writing $\epsilon$ for the evaluation map associated with $C^D$):
  \[
    \Xi\prn{f} = \Prod{d \in \Hom{\ObjTerm{}}{D}} \Psi\prn{d} \to \Phi\prn{\epsilon\gls{f,d}}
  \]
\end{exa}

Informally, therefore, we view $\GL{\Mor[F]{\CC}{\DD}}$ as the category of $\DD$-valued predicates
on $\CC$ and the construction of exponentials, products, \etc{} within $\GL{F}$ corresponds to
defining a logical relation on $\CC$. See Mitchell and Scedrov~\cite{mitchell:1993} for an
exposition on this perspective.

Carrying out a normalization-by-gluing proof, therefore, turns the classical approach on its
head. Originally one defined the normalization algorithm then showed it to be sound, complete, and
idempotent. When carrying out the proof by gluing, the algorithm is not defined up front. Instead,
one carefully constructs a gluing category $\GL{F}$ built on a functor out of the category of
contexts of the initial model $\mathcal{I}$. Concretely, this is the category of syntactic contexts
and simultaneous substitutions between them up to definitional equality. The heart of the argument
then breaks down into three steps:

\begin{enumerate}
\item We show that $\GL{F}$ supports a particular model of type theory $\mathcal{G}$.
\item We define a \emph{reify} operation which sends terms from $\mathcal{G}$ to normal forms.
\item We show that the projection $\pi_0$ induces a morphism of models
  $\Mor{\mathcal{G}}{\mathcal{I}}$ and that for a given term $x$ in $\mathcal{G}$ reifying $x$
  yields a normal form for $\pi_0\prn{x}$.
  \label{item:strict-equations}
\end{enumerate}

In particular, types in $\mathcal{G}$ will be chosen such that they consist of a type from the
initial model along with a proof-relevant predicate carving out those terms which have (suitably
hereditary) normal forms. A term in this model is then a term from the syntactic model together with
a witness for the proof-relevant predicate associated with the type.

The first step and the universal property of the initial model produces a morphism of models
$\Mor[i]{\mathcal{I}}{\mathcal{G}}$ and the second step ensures that $\pi_0 \circ i =
\ArrId{}$. Remarkably, this already defines a sound and complete normalization algorithm. The
algorithm simply takes a syntactic term $M : A$, regards it as an element of the initial model, and
then reifies $i\prn{M}$ to obtain the normal form. Moreover, because $\pi_0 \circ i = \ArrId{}$ we
conclude that this yields a normal form for the supplied $M$.

To a coarse approximation, the construction of $\mathcal{G}$ and reification specifies the
normalization algorithm and proves its soundness in a single step. The attentive reader will notice,
however, that the completeness requirement from Section~\ref{sec:introduction:normalization} seems
to be absent from this new story. In fact, in this approach completeness is automatic and no proof
is required. Indeed, terms and types within the initial model are realized by equivalences classes
of syntactic terms and types taken up to definitional equality. Accordingly, the morphism $i$---and
therefore the normalization algorithm---cannot distinguish between definitional equal terms.

One might suspect that working with equivalence classes of terms when defining $\mathcal{G}$ simply
causes the burden to shift so that---while there is no need to prove completeness separately---the
work of such a proof is spread throughout the construction of $\mathcal{G}$. In fact the opposite is
the case: working with terms up to definitional equality substantially simplifies the construction
of $\mathcal{G}$. Connectives in type theory only have universal properties up to definitional
equality. Only when working with equivalences classes therefore, can we use these universal
properties and benefit from existing results. For instance, we shall see that our construction of
dependent products in our gluing model is essentially mechanical.

The gluing approach yields other unexpected advantages. Recall that $\GL{F}$ intuitively consists of
\emph{proof-relevant} predicates. This proof relevance is crucial to an elegant treatement of
universes in the model~\cite{coquand:2019}. We are able to define the predicate associated with an
element of a universe to consist not only of an appropriate normal form but to also contain the data
of the type it encodes within the model. In proof-irrelevant settings, universes were a frequent
source of difficulty which necessitated laborious techniques to encode~\cite{allen:1987}.

\subsection{Synthetic Tait computability}
\label{sec:introduction:stc}

Using gluing to prove normalization is certainly an improvement over `free-hand' proofs of
normalization-by-evaluation, but the picture is not as rosy as it may first appear. Models of type
theory are subject to a variety of strict equations (see Item \ref{item:strict-equations} on
page \pageref{item:strict-equations}) which often
force external constructions, where naturality obligations can be prohibitive. Worse, the passage
between mathematics internal to the gluing category and external constructions is difficult
and the boundary frequently raises mismatches.

We follow Sterling and Harper~\cite{sterling:modules:2021} and adopt a synthetic approach to
gluing. We begin with two crucial observations. First, while models of type theory are strangely
behaved objects, one can often embed a model into a presheaf topos and thereby work in an extremely
rich setting. Second, when gluing together presheaf topoi along a nice functor
$\GL{\Mor[F]{\PSH{\CC}}{\PSH{\DD}}}$, the result is another presheaf topos and the internal language
of this topos contains lex idempotent monads ($\Open,\Closed$) allowing one to recover both
$\PSH{\CC}$ and $\PSH{\DD}$.

Sterling and collaborators have then shown that it is possible to work exclusively within the
internal language of $\GL{F}$ to construct the normalization model and have termed this approach
\emph{synthetic Tait computability (STC)}. Experience has shown that working internally simplifies
constructions involved in the gluing model, making it practical to prove metatheorems for even
extremely complex type theories like cubical type
theory~\cite{sterling:modules:2021,sterling:2021,sterling:phd,gratzer:lob:2022,sterling:2022}.

Proofs using STC construct the model within $\GL{F}$ by defining a sequence of constants within the
internal language. Accordingly, the heart of the normalization proof is realized by a series of
programming exercises in extensional type theory. This alone does not remove the strict equations
that cause trouble with typical gluing proofs but it does provide a systematic approach to handling
them. Concretely, within an STC proof, all the required strict equations have a particular form: for
some type operator in the object theory, we are given an element $\mathsf{op} : \Open \Ty{}$
corresponding to the operator in the syntactic model, and we must extend this to an element of
$\Ty{}$. Within the internal language, the two components of this problem (the element of $\Ty{}$
and the proof that it extends $\mathsf{op}$) can be represented by an element of the following
dependent sum:\footnote{Here we have used standard syntactic sugar to represent the monadic
  operations of $\Open$.}
\[
  \Sum{A : \Ty{}} x \gets \mathsf{op}; \Open \prn{A = x}
\]
The second component in particular represents the aforementioned strict equation. In practice, it
is easy to obtain an element of $\Ty{}$ which extends $\mathsf{op}$ up to isomorphism \ie{} an
element of the following type:
\[
  \Sum{A : \Ty{}} x \gets \mathsf{op}; \Open \prn{\Tm{}\prn{A} \cong \Tm{}\prn{x}}
\]

Remarkably, this proves to be enough. The internal language of $\GL{F}$ supports a strictification
axiom~\cite{orton:2018} which provides a section to the canonical projection from the first type to
the second. We are therefore able to construct various connectives which agree only up to
isomorphism with their syntactic counterparts and correct them to construct the model. For instance,
a dependent product is determined by a universal property and it is possible to construct a type in
$\GL{F}$ with this property by virtue of general categorical theorems. However, the result will only
satisfy the required equation up to isomorphism. The strictification axiom allows STC proofs to
benefit from the general categorical result without resorting to unfolding the construction supplied
by the abstract argument.

\subsubsection{Synthetic Tait computability for \MTT{}}
Unlike Martin-L{\"o}f type theory or cubical type theory, a model of \MTT{} is not a single category
equipped with additional structure. Rather, a model is a network of categories, each supporting
their own individual model of type theory which are then connected by various adjoints and natural
transformations. The internal language of any of these categories is insufficient to construct the
gluing model, so it is necessary to generalize from working in the extensional type theory of a
topos to working in all topoi simultaneously using extensional \MTT{}. Each topos then comes
equipped with the structure of STC: a pair of lex monads and a strictification axiom. We prove that
this mode-local structure is respected by the \MTT{} modalities between topoi and call the resulting
language \emph{multimodal synthetic Tait computability}. The smooth interaction between MTT
modalities and the lex monads $\Open$ and $\Closed$ ensures that the key techniques of STC proofs
can be generalized to multimodal STC.

With this machinery, we are able to give a concise and conceptual construction of the gluing model
and extract the first normalization algorithm for multimodal type theory. In practice, this internal
proof is necessary; removing the simplifying assumption on substitutions used in the canonicity
proof given by Gratzer et al.~\cite{gratzer:mtt-journal:2021} is already nearly intractable.

\subsection{Contributions}

We contribute a normalization algorithm for \MTT{} equipped with the full suite of connectives:
dependent sums, products, booleans, intensional identity types, a universe, and modal types. In
addition to the usual corollaries of normalization (decidability of type checking, injectivity of
type constructors, \etc{}), this sharpens the canonicity result of Gratzer et
al.~\cite{gratzer:2020}. This algorithm applies to any choice of mode theory and therefore
simultaneously establishes normalization results for many specialized modal calculi.

In order to prove this result, we advance modern gluing techniques to apply to modal type theories
and demonstrate that extensional \MTT{} itself is a suitable metalanguage for carrying out the proof
of normalization-by-gluing. We further argue that these techniques scale by extending the proof to a
version of \MTT{} supplemented with crisp induction principles and deduce that \eg{}, normalization
continues to hold.

Section~\ref{sec:mtt} gives a brief tutorial on \MTT{} and introduces normal forms for this type
theory. In Section~\ref{sec:models}, we discuss the models of \MTT{} and relax the definition of a
model of \MTT{} to obtain \emph{\MTT{} cosmoi}. We prove that the syntactic cosmos enjoys a
privileged position among \MTT{} cosmoi
(Theorem~\ref{thm:cosmoi:quasi-initiality}). Section~\ref{sec:mstc} introduces \emph{multimodal
  synthetic Tait computability} and shows that gluing together a network of topoi results in a model
of extensional \MTT{} equipped with STC structure in each mode
(Theorem~\ref{thm:mstc:fundamental}). Finally, in Section~\ref{sec:normalization-cosmos} we
construct the normalization cosmos (Theorem~\ref{thm:model:mtt-cosmos}) and extract the
normalization function in Section~\ref{sec:normalization}
(Theorem~\ref{thm:normalization:normalization}). Section~\ref{sec:crisp} discusses an extension of
this proof to support crisp induction.

\section{A primer on \MTT{}}
\label{sec:mtt}

We collect the key ideas of \MTT{}~\cite{gratzer:mtt-journal:2021}. First, as mentioned in
Section~\ref{sec:introduction}, \MTT{} is parametrized by a mode theory: a strict 2-category $\Mode$
whose objects are modes, morphisms are modalities, and 2-cells are natural transformations between
modalities. Henceforth, we will work with \MTT{} over a fixed mode theory $\Mode$.

\MTT{} plays two distinct roles in this paper. First, it is the object theory under consideration
and the subject of our normalization theorem. However, as the proof of normalization uses \MTT{} as
an internal language to construct the normalization model \MTT{} is also used as a metalanguage.
These two different uses invite two very distinct perspectives on the type theory. In order to
crystallize \MTT{} precisely enough for the normalization result, we will view \MTT{} as a
particular generalized algebraic theory (GAT). Accordingly, binding is handled by De Bruijn indices
and the theory uses explicit substitutions~\cite{martin-lof:1992}. On the other hand, we will not
use De Brujin indices and explicit substitutions when working with \MTT{} as a metalanguage. In
these instances, we will treat \MTT{} as a normal type theory and avail ourselves of conveniences
similar to what a proof assistant like Agda might provide.

As a compromise, we introduce \MTT{} in Sections~\ref{sec:primer:mode-local} and
\ref{sec:primer:modalities} as a formal theory but go through several important constructions in
Section~\ref{sec:primer:examples} using the informal surface-language employed by much of
Section~\ref{sec:normalization-cosmos}. For a comprehensive account of both perspectives, we refer
the reader to Gratzer et al.~\cite{gratzer:mtt-journal:2021}.

\subsection{Mode-local connectives in \MTT{}}
\label{sec:primer:mode-local}

Each mode in \MTT{} constitutes its own separate type theory. In fact, each mode $m$ is equipped
with its own copy the of judgments of type theory \eg{}, $\IsCx{\Gamma}$, $\IsTy{A}$,
$\IsTm{M}{A}$. Much of the theory of \MTT{} is \emph{mode-local} and only mentions a single copy of
these judgments at a time. For these connectives the rules are precisely the standard rules from
\MLTT{}, replicated for each mode. The connectives of type theory---dependent sums, intensional
identity types, booleans---are all incorporated in this fashion. Each mode also contains a
\emph{weak} universe {\`a la Tarski}. Explicitly, this means that there are separate codes and an
$\Dec{-}$ operation decoding a code to a type, but the decoding operation only commutes with
connectives up to isomorphism. While the restriction to weak universes is not fundamental, it
simplifies the proof and recent implementations have shown them to be practical~\cite{cooltt}.

\subsection{Modalities in \MTT{}}
\label{sec:primer:modalities}

The novelty of \MTT{} comes from those connectives which mix two modes: the modalities.  \MTT{}
draws inspiration from Fitch-style type theories~\cite{clouston:fitch:2018,birkedal:2020} and
defines each modality together with an adjoint action on contexts. Accordingly, each
$\Mor[\mu]{n}{m}$ defines a context former sending contexts in mode $m$ to contexts in mode $n$ and
this is then used to define modal types $\Modify{A}$:
\begin{mathparpagebreakable}
  \inferrule{
    \IsCx{\Gamma}
  }{
    \IsCx{\LockCx{\Gamma}}<n>
  }
  \and
  \inferrule{
    \IsTy[\LockCx{\Gamma}]{A}<n>
  }{
    \IsTy{\Modify{A}}
  }
  \and
  \inferrule{
    \IsTm[\LockCx{\Gamma}]{M}{A}<n>
  }{
    \IsTm{\MkBox{M}}{\Modify{A}}
  }
\end{mathparpagebreakable}

These context operations assemble into a 2-functor $m \mapsto \Cx{m}$ from $\Coop{\Mode}$ to the
category of categories, selecting the various categories of contexts.\footnote{Given a 2-category $\CC$, recall that $\Coop{\CC}$ is a 2-category
  with the same objects as $\CC$ but with 1- and 2-cells reversed.} Concretely, a substitution
$\IsSb[\Delta]{\gamma}{\Gamma}$ lifts to a substitution
$\IsSb[\LockCx{\Delta}]{\LockSb{\gamma}}{\LockCx{\Gamma}}<n>$ and each 2-cell
$\Mor[\alpha]{\nu}{\mu}$ induces a substitution
$\IsSb[\LockCx{\Gamma}]{\Key{\alpha}{}}{\LockCx{\Gamma}<\nu>}<n>$. These operations satisfy several
equations to organize them into a 2-functor \eg{},
$\EqSb[\LockCx{\Gamma}]{\LockSb{\ISb}}{\ISb}{\LockCx{\Gamma}}<n>$ and
$\EqCx{\LockCx{\LockCx{\Gamma}}<\xi>}{\LockCx{\Gamma}<\mu\circ\xi>}<o>$. We record these rules in
Figure~\ref{fig:mtt:substitutions}.

\begin{figure}
  \begin{mathpar}
    \inferrule{
      \Mor[\mu]{n}{m}
      \\
      \IsCx{\Gamma}
    }{
      \IsCx{\LockCx{\Gamma}}<n>
    }
    \and
    \inferrule{
      \Mor[\mu]{n}{m}
      \\
      \IsSb{\delta}{\Delta}
    }{
      \IsSb[\LockCx{\Gamma}]{\LockSb{\delta}}{\LockCx{\Delta}}<n>
    }
    \and
    \inferrule{
      \Mor[\mu]{n}{m}
      \\
      \IsSb[\Gamma]{\delta_0}{\Delta_0}
      \\
      \IsSb[\Delta_0]{\delta_1}{\Delta_1}
    }{
      \EqSb[\LockCx{\Gamma}]{
        \LockSb{\prn{\delta_1 \circ \delta_0}}
      }{
        \LockSb{\delta_1} \circ \LockSb{\delta_0}
      }{
        \LockCx{\Delta_1}
      }<n>
    }
    \and
    \inferrule{
      \Mor[\mu]{n}{m}
      \\
      \IsCx{\Gamma}
    }{
      \EqSb[\LockCx{\Gamma}]{
        \LockSb{\ISb}
      }{
        \ISb
      }{
        \LockCx{\Gamma}
      }<n>
    }
    \and
    \inferrule{
      \Mor[\nu]{o}{n}
      \\
      \Mor[\mu]{n}{m}
      \\
      \IsCx{\Gamma}
    }{
      \EqCx{\LockCx{\Gamma}<\mu\circ\nu>}{\LockCx{\LockCx{\Gamma}}<\nu>}<o>
    }
    \and
    \inferrule{
      \Mor[\nu]{o}{n}
      \\
      \Mor[\mu]{n}{m}
      \\
      \IsSb{\delta}{\Delta}
    }{
      \EqSb[\LockCx{\Gamma}<\mu\circ\nu>]{
        \LockSb{\LockSb{\delta}}<\nu>
      }{
        \LockSb{\delta}<\mu\circ\nu>
      }{
        \LockCx{\Delta}<\mu\circ\nu>
      }<o>
    }
    \and
    \inferrule{
      \Mor[\mu,\nu]{n}{m}
      \\
      \Mor[\alpha]{\nu}{\mu}
      \\
      \IsSb{\delta}{\Delta}
    }{
      \IsSb[\LockCx{\Gamma}<\mu>]{
        \Key{\alpha}{\Gamma}
      }{
        \LockCx{\Gamma}<\nu>
      }<n>
    }
    \and
    \inferrule{
      \Mor[\mu]{n}{m}
      \\
      \IsCx{\Gamma}
    }{
      \EqSb[\LockCx{\Gamma}]{\ISb}{\Key{\ArrId{}}{\Gamma}}{\LockCx{\Gamma}}<n>
    }
    \and
    \inferrule{
      \IsCx{\Gamma,\Delta}
      \\
      \Mor[\mu,\nu]{n}{m}
      \\
      \IsSb{\delta}{\Delta}
      \\
      \Mor[\alpha]{\nu}{\mu}
    }{
      \EqSb[\LockCx{\Gamma}]{
        \Key{\alpha}{\Gamma} \circ \prn{\LockSb{\delta}}
      }{
        \prn{\LockSb{\delta}<\nu>} \circ \Key{\alpha}{\Delta}
      }{
        \LockCx{\Delta}<\nu>
      }<n>
    }
    \and
    \inferrule{
      \IsCx{\Gamma}<m>
      \\
      \Mor[\mu_0, \mu_1, \mu_2]{n}{m}
      \\
      \Mor[\alpha_0]{\mu_0}{\mu_1}
      \\
      \Mor[\alpha_1]{\mu_1}{\mu_2}
    }{
      \EqSb[
        \LockCx{\Gamma}<\mu_2>
      ]{
        \Key{\alpha_1 \circ \alpha_0}{\Gamma}
      }{
        \Key{\alpha_0}{\Gamma} \circ \Key{\alpha_1}{\Gamma}
      }{
        \LockCx{\Gamma}<\mu_0>
      }<n>
    }
    \and
    \inferrule{
      \IsCx{\Gamma}
      \\
      \Mor[\nu_0, \nu_1]{o}{n}
      \\
      \Mor[\mu_0, \mu_1]{n}{m}
      \\
      \Mor[\beta]{\nu_0}{\nu_1}
      \\
      \Mor[\alpha]{\mu_0}{\mu_1}
    }{
      \EqSb[
        \LockCx{\Gamma}<\mu_1 \circ \nu_1>
      ]{
        \Key{\alpha \HComp \beta}{\Gamma}
      }{
        \LockSb{\Key{\alpha}{\Gamma}}<\nu_0> \circ \Key{\beta}{\LockCx{\Gamma}<\mu_1>}
      }{
        \LockCx{\Gamma}<\mu_0 \circ \nu_0>
      }<o>
    }
  \end{mathpar}
  \caption{Key rules for contexts and substitutions in \MTT{}}
  \label{fig:mtt:substitutions}
\end{figure}

Two basic questions remain: what is the elimination principle for $\Modify{A}$ and which terms can
be constructed in the context $\LockCx{\Gamma}$? Both of these problems are addressed through the
same idea, the final component of \MTT{}. We generalize the context extension $\OldECx{\Gamma}{A}$
from \MLTT{} to annotate each variable with a modality:
\[
  \inferrule{
    \IsCx{\Gamma}
    \\
    \IsTy[\LockCx{\Gamma}]{A}<n>
  }{
    \IsCx{\ECx{\Gamma}{A}}
  }
\]
Intuitively, $\ECx{\Gamma}{A}$ plays the same role as $\OldECx{\Gamma}{\Modify{A}}$ and comes
equipped with a similar universal property: a substitution $\IsSb[\Delta]{\gamma}{\ECx{\Gamma}{A}}$
is precisely determined by a substitution $\IsSb[\Delta]{\gamma'}{\Gamma}$ and a term
$\IsTm[\LockCx{\Delta}]{M}{\Sb{A}{\LockSb{\gamma'}}}<n>$. The ordinary context extension $\Gamma.A$
is recovered by taking $\mu = \ArrId{}$; the equation $\LockCx{\Gamma}<\ArrId{}> = \Gamma$ ensures
that the universal properties of $\Gamma.A$ and $\ECx{\Gamma}{A}<\ArrId{}>$ match.

Despite the similarities between $\ECx{\Gamma}{A}$ and $\ECx{\Gamma}{\Modify{A}}<\ArrId{}>$, they
occupy different positions in the theory. The variable rule of \MTT{} is adjusted to take into
account modal annotations and require that the modalities in the context must cancel a variable's
annotation:
\[
  \inferrule{
    \IsCx{\Gamma}
    \\
    \IsTy[\LockCx{\Gamma}]{A}<n>
  }{
    \IsTm[\LockCx{\ECx{\Gamma}{A}}]{\Var{0}}{\Sb{A}{\LockSb{\Wk}}}<n>
  }
\]

As in Martin-L{\"o}f type theory, it is necessary to apply a weakening substitution $\Wk$ to $A$
when describing the type of $\Var{0}$. The normal variable rule arises again as a special case after
setting $\mu = \ArrId{}$. Note that attempting to state such a variable rule for
$\ECx{\Gamma}{\Modify{A}}<\ArrId{}>$ would quickly introduce issues around substitution within the
theory, so these two contexts behave quite differently in practice.

\begin{rem}
  From the view of Fitch-style type theories where $\LockCx{-}$ is left adjoint to the modal type,
  this rule plays the role of the counit; it allows us to pass from $L(R(A))$ to $A$.
\end{rem}

The addition of modal annotations creates a redundancy in our system: we may hypothesize of
$\Modify[\mu]{A}$ with annotation $\nu$ or directly hypothesize over $A$ with annotation
$\nu\circ\mu$. There is a substitution navigating in one direction, but not the other:
\[
  \IsSb[\ECx{\Gamma}{A}<\nu\circ\mu>]{\ESb{\Wk}{\MkBox[\mu]{\Var{0}}}}{\ECx{\Gamma}{\Modify[\mu]{A}}<\nu>}<o>
\]

This mismatch is addressed through elimination for $\Modify[\mu]{-}$. Informally, this rule ensures
that these two contexts are isomorphic `from the perspective of a type':\footnote{Formally, this
  rule ensures that, among others, this map is anodyne in the sense of Awodey~\cite{awodey:2018}.}
\begin{mathparpagebreakable}
  \inferrule{
    \Mor[\nu]{m}{o}
    \\
    \Mor[\mu]{n}{m}
    \\\\
    \IsCx{\Gamma}<o>
    \\
    \IsTy[\LockCx{\LockCx{\Gamma}<\nu>}<\mu>]{A}<n>
    \\
    \IsTy[\ECx{\Gamma}{\Modify[\mu]{A}}<\nu>]{B}
    \\
    \IsTm[\LockCx{\Gamma}<\nu>]{M_0}{\Modify[\mu]{A}}<m>
    \\
    \IsTm[\ECx{\Gamma}{A}<\nu\circ\mu>]{M_1}{\Sb{B}{\ESb{\Wk}{\MkBox[\mu]{\Var{0}}}}}<o>
  }{
    \IsTm{\LetMod{M_0}{M_1}<\nu>[\mu]}{\Sb{B}{\ESb{\ISb}{M_0}}}<o>
  }
  \and
  \LetMod{\MkBox[\nu]{M_0}}{M_1}<\nu>[\mu] = \Sb{M_1}{\ESb{\ISb}{M_0}}
\end{mathparpagebreakable}

Notice that the elimination rule for the modal type $\Modify[\mu]{-}$ is parameterized by an
additional modality $\nu$. We refer to $\mu$ as the \emph{main modality} and $\nu$ as the
\emph{framing modality}.

\begin{rem}
  Fitch-style type theories require
  $\IsSb[\ECx{\Gamma}{A}<\nu\circ\mu>]{\ESb{\Wk}{\MkBox[\mu]{\Var{0}}}}{\ECx{\Gamma}{\Modify[\mu]{A}}<\nu>}<o>$
  to be invertible. Such an inverse, however, again disrupts substitution in the presence of
  multiple modalities. For an extended discussion of this point and various potential solutions, see
  Gratzer et al.~\cite{gratzer:fitchtt:2022}.
\end{rem}

In addition to modal types, dependent products in \MTT{} are also modalized so that $A \to B$ is
replaced by $\Fn{A}{B}$:
\begin{mathparpagebreakable}
  \inferrule{
    \IsTm[\ECx{\Gamma}{A}]{M}{B}
  }{
    \IsTm{\Lam{M}}{\Fn{A}{B}}
  }
  \and
  \inferrule{
    \IsTm{M}{\Fn{A}{B}}
    \\
    \IsTm[\LockCx{\Gamma}]{N}{A}<n>
  }{
    \IsTm{\App{M}{N}}{\Sb{B}{\ESb{\ISb}{N}}}
  }
\end{mathparpagebreakable}

This feature is a useful convenience; it ensures that many functions avoid the need to accept an
argument of modal type only to immediately apply the elimination rule. We will see frequent examples
of this pattern later as \MTT{} is used as a metalanguage.

\subsection{Standard combinators within \MTT{}}
\label{sec:primer:examples}

\NewDocumentCommand{\Triv}{}{\mathsf{triv}}
\NewDocumentCommand{\Comp}{}{\mathsf{comp}}

As the assignment $\Gamma \mapsto \LockCx{\Gamma}$ is pseudofunctorial, its adjoint action on types is
likewise functorial \emph{up to propositional equality}. In particular, there are equivalences
$\Triv : \Modify[\ArrId{}]{A} \to A$ and
$\Comp : \Modify{\Modify[\nu]{A}} \to \Modify[\mu\circ\nu]{A}$:

\begin{align*}
  &\Triv\prn{x} = \LetMod{x}[y]{y}<\ArrId{}>[\ArrId{}]
  \\
  &\Triv^{-1}\prn{x} = \MkBox[\ArrId{}]{x}
  \\[0.1cm]
  &\Comp\prn{x} = \LetMod{x}[y_0]{\LetMod{y_0}[y_1]{\MkBox[\mu\circ\nu]{y_1}}<\nu>[\mu]}<\mu>[\ArrId{}]
  \\
  &\Comp^{-1}\prn{x} = \LetMod{x}[y]{\MkBox[\mu]{\MkBox[\nu]{y}}}<\mu\circ\nu>[\ArrId{}]
\end{align*}

Each modality $\Modify[\mu]{-}$ also satisfies the modal principle referred to as \emph{axiom K}
\ie{}, they preserve finite products. In practice, this property serves as an internalization of
functoriality as it provides a canonical comparison map
$\Modify{A \to B} \to \Modify{A} \to \Modify{B}$. In fact, we can prove a dependent version of this
map as in Birkedal et al.~\cite{birkedal:2020}:
\begin{align*}
  &\prn{\ZApp{}} :
  \Modify[\mu]{\prn{x : A} \to B\prn{x}}
  \to \prn{a : \Modify[\mu]{A}}
  \to \LetMod{a}[a_0]{\Modify{B\prn{a_0}}}[1]
  \\
  &f \ZApp{} a = \LetMod{f}[f_0]{\LetMod{a}[a_0]{\MkBox{f_0\prn{a_0}}}[1]}[1]
\end{align*}
In functional programming parlance, modalities are \emph{applicative functors} though without an
operation $A \to \Modify{A}$~\cite{mcbride:2008}.

While it is far less useful, one can also define a version of $\ZApp{}$ using the modalized
dependent product rather than accepting elements of $\Modify{-}$:
\begin{align*}
  &\prn{\ZApp{}'} :
  \DeclNameless{\prn{x : A} \to B\prn{x}}
  \to \DeclVar{a}{A}
  \to \Modify{B\prn{a}}
  \\
  &f \mathbin{\ZApp{}'} a = \MkBox{f\prn{a}}
\end{align*}

This is indicative of a common pattern; it is typically far more concise to use the modalized
dependent product instead of accepting $\Modify{-}$ in order to avoid needing to immediately eliminate
arguments.

\subsection{Normal and neutral forms in \MTT{}}
\label{sec:primer:normals-and-neutrals}

As mentioned in Section~\ref{sec:introduction:normalization}, the starting point for normalization
is the definition of normal form. In \MTT{}---as in other type theories---normal forms are presented
together with a class of neutral forms. Intuitively, normal forms capture terms in $\beta$-normal
and $\eta$-long form while neutrals are chains of eliminations applied to a variable.

We define normal and neutral forms as separate syntactic classes, equipped with their own family of
typing judgments and decoding functions sending them to terms. Dependency complicates this
definition as various typing rules require substitution in the types of premises or the
conclusion. Unfortunately, it is just as hard to define substitution on normal forms as it is to
define normalization in general~\cite{watkins:2004}. Accordingly, a normal form (resp. neutral,
normal type) is typed by the judgment $\IsNf{u}{A}$ (resp. $\IsNe{e}{A}$, $\IsNfTy{\tau}$) where $A$
is not required to be any sort of normal form. Furthermore, these judgments are defined
inductive-recursively with decoding functions $\DecNf{u}$ (resp. $\DecNe{e}$, $\DecNfTy{\tau}$)
which send a normal form (resp. neutral, normal type) to its corresponding piece of syntax.  Normal
and neutral forms for mode-local connectives are unchanged from their standard presentation in type
theory:
\begin{grammar}
  Normals & u &
  \NfLam{u} \GrmSep{} \NfInj{e} \GrmSep{} \NfMkBox{u} \GrmSep{} \dots
  \\
  Neutral & e &
  \NeVar{k}{\alpha} \GrmSep{} \NeApp{e}{u} \GrmSep{} \NeLetMod{\mu}{\nu}{\tau}{e}{u} \GrmSep{} \dots
  \\
  Normal types & \tau &
  \NfFn{\tau}{\sigma} \GrmSep{} \NfModify{\tau} \GrmSep{} \NfDec{u} \GrmSep{} \dots
\end{grammar}

We defer a more complete presentation of the judgments and decoding function to
Figure~\ref{fig:app:normals-and-neutrals}, but remark that the neutral form for variables is
annotated with a 2-cell and index, decoding to $\Var{0}$ together with a combination of weakening
and 2-cell substitutions $\Wk$ and $\Key{\alpha}{}$. Note that we require that $\Dec{-}$ commute
with type formers only up to isomorphism (weak Tarski universes) we must include neutral and normal
forms for \eg{}, $\Dec{\ModifyCode{A}}$ as well as other type connectives. We include only those for
$\ModifyCode{-}$ as they are representative of the general pattern.

\begin{figure*}
  \begin{mathparpagebreakable}
    \inferrule{ }{
      \IsRen{\EmpRen}{\EmpTele}
      \\
      \DecRen{\EmpRen} = \EmpSb
    }
    \and
    \inferrule{ }{
      \IsRen[\ECx{\Gamma}{A}]{\WkRen}{\Gamma}
      \\
      \DecRen{\WkRen} = \Wk
    }
    \and
    \inferrule{ }{
      \IsRen{\IRen}{\Gamma}
      \\
      \DecRen{\IRen} = \ISb
    }
    \and
    \inferrule{
      \IsRen[\Gamma_0]{r}{\Gamma_1}
      \\
      \IsRen[\Gamma_1]{s}{\Gamma_2}
    }{
      \IsRen[\Gamma_0]{s \circ r}{\Gamma_2}
      \\
      \DecRen{s \circ r} = \DecRen{s} \circ \DecRen{r}
    }
    \and
    \inferrule{
      \IsRen{r}{\Delta}
    }{
      \IsRen[\LockCx{\Gamma}]{\LockRen{r}}{\LockCx{\Delta}}<n>
      \\
      \DecRen{\LockRen{r}} = \LockSb{\DecRen{r}}
    }
    \and
    \inferrule{
      \Mor[\mu, \nu]{n}{m}
      \\
      \Mor[\alpha]{\nu}{\mu}
    }{
      \IsRen[\LockCx{\Gamma}]{\KeyRen{\alpha}{\Gamma}}{\LockTele{\Gamma}<\nu>}<n>
      \\
      \DecRen{\KeyRen{\alpha}{\Gamma}} = \Key{\alpha}{\Gamma}
    }
    \and
    \inferrule{
      \IsRen{r}{\Delta}<m>
      \\
      \IsNe[\LockCx{\Gamma}]{\NeVar{k}{\alpha}}{\Sb{A}{\LockSb{\DecRen{r}}}}<n>
    }{
      \IsRen{\ERen{r}{\NeVar{k}{\alpha}}}{\ECx{\Delta}{A}}
      \\
      \DecRen{\ERen{r}{\NeVar{k}{\alpha}}} = \ESb{\DecRen{r}}{\DecNe{\NeVar{k}{\alpha}}}
    }
  \end{mathparpagebreakable}
  \caption{Complete definition of renamings}
  \label{fig:app:renamings}
\end{figure*}

\begin{figure*}
  \begin{mathparpagebreakable}
    \inferrule{ }{
      \IsNfTy{\NfBool}
      \and
      \IsNfTy{\NfUni}
    }
    \and
    \inferrule{
      \IsNfTy[\Gamma]{\tau}
      \\
      \IsNfTy[\ECx{\Gamma}{\DecNfTy{\tau}}]{\sigma}
    }{
      \IsNfTy{\NfFn{\tau}{\sigma}}
    }
    \and
    \inferrule{
      \IsNfTy{\tau}
      \\
      \IsNfTy[\ECx{\Gamma}{\DecNfTy{\tau}}<\ArrId{}>]{\sigma}
    }{
      \IsNfTy{\NfProd{\tau}{\sigma}}
    }
    \and
    \inferrule{
      \IsNfTy{\tau}
      \\
      \IsNf{u,v}{\DecNfTy{\tau}}
    }{
      \IsNfTy{\NfId{\tau}{u}{v}}
    }
    \and
    \inferrule{
      \IsNfTy[\LockCx{\Gamma}]{\tau}<n>
    }{
      \IsNfTy{\NfModify{\tau}}
    }
    \and
    \inferrule{
      \IsNf{u}{\NfUni}
    }{
      \IsNfTy{\NfDec{u}}
    }
  \end{mathparpagebreakable}
  \begin{mathparpagebreakable}
    \inferrule{
      \Gamma(k) = \DeclNameless{A}
      \\
      \Locks{\Gamma}{k} = \nu
      \\
      \Mor[\alpha]{\mu}{\nu}
    }{
      \IsNe{\NeVar{k}{\alpha}}{
        \Sb{A}{\Key{\alpha}{} \circ (\LockSb{\Wk}<\nu_{k-1}>) \dots \circ (\LockSb{\Wk}<\nu_{0}>)}
      }
    }
    \and
    \inferrule{ }{
      \IsNf{\NfTrue}{\Bool}
      \\
      \IsNf{\NfFalse}{\Bool}
    }
    \and
    \inferrule{
      \IsNe{e}{\Bool}
    }{
      \IsNf{\NfInj{e}}{\Bool}
    }
    \and
    \inferrule{
      \IsNfTy[\ECx{\Gamma}{\Bool}<\ArrId{m}>]{\tau}
      \\
      \IsNe{e}{\Bool}
      \\
      \IsNf{v_1}{\Sb{\DecNfTy{\tau}}{\ESb{\ISb}{\True}}}
      \\
      \IsNf{v_2}{\Sb{\DecNfTy{\tau}}{\ESb{\ISb}{\False}}}
    }{
      \IsNe{\NeBoolRec{\tau}{e}{v_1}{v_2}}{\Sb{\DecNfTy{\tau}}{\ESb{\ISb}{\DecNe{e}}}}
    }
    \and
    \inferrule{
      \IsNf{u}{A}
    }{
      \IsNf{\NfRefl{u}}{\Id{A}{\DecNf{u}}{\DecNf{u}}}
    }
    \and
    \inferrule{
      \IsTm{M_0,M_1}{A}
      \\
      \IsNe{e}{\Id{A}{M_0}{M_1}}
    }{
      \IsNf{\NfInj{e}}{\Id{A}{M_0}{M_1}}
    }
    \and
    \inferrule{
      \IsTm{M_0,M_1}{A}
      \\
      \IsNe{e}{\Id{A}{M_0}{M_1}}
      \\
      \IsNfTy[
        \ECx{\ECx{\ECx{\Gamma}{A}<\ArrId{m}>}{A}<\ArrId{m}>}{\Id{\Sb{A}{\Wk[2]}}{\Var{1}}{\Var{0}}}<\ArrId{m}>
      ]{\tau}
      \\
      \IsNf[\ECx{\Gamma}{A}<\ArrId{}>]{u}{\Sb{\DecNfTy{\tau}}{\ESb{\ESb{\ESb{\ISb}{\Var{0}}}{\Var{0}}}{\Refl{\Var{0}}}}}
    }{
      \IsNe{\NeIdRec{\tau}{u}{e}}{\Sb{\DecNfTy{\tau}}{\ESb{\ESb{\ESb{\ISb}{M_0}}{M_1}}{P}}}
    }
    \and
    \inferrule{
      \IsNf[\ECx{\Gamma}{A}]{u}{B}
    }{
      \IsNf{\NfLam{u}}{\Fn{A}{B}}
    }
    \and
    \inferrule{
      \IsNe{e}{\Fn{A}{B}}
      \\
      \IsNf{u}{A}
    }{
      \IsNe{\NeApp{e}{u}}{\Sb{B}{\ESb{\ISb}{\DecNf{u}}}}
    }
    \and
    \inferrule{
      \IsNf[\LockCx{\Gamma}]{u}{A}<n>
    }{
      \IsNf{\NfMkBox{u}}{\Modify{A}}
    }
    \and
    \inferrule{
      \IsNe{e}{\Modify{A}}
    }{
      \IsNf{\NfInj{e}}{\Modify{A}}
    }
    \and
    \inferrule{
      \IsNe[\LockCx{\Gamma}]{u}{\Modify[\nu]{A}}<n>
      \\
      \IsNfTy[\ECx{\Gamma}{\Modify[\nu]{A}}]{\tau}
      \\
      \IsNf[\ECx{\Gamma}{A}<\mu \circ \nu>]{u}{\Sb{\DecNfTy{\tau}}{\ESb{\Wk}{\MkBox[\nu]{\Var{0}}}}}
    }{
      \IsNe{\NeLetMod{\mu}{\nu}{\tau}{e}{u}}{\Sb{\DecNfTy{\tau}}{\ESb{\ISb}{\DecNf{u}}}}
    }
    \and
    \inferrule{
      \IsNe{e}{\Uni}
    }{
      \IsNf{\NfInj{e}}{\Uni}
    }
    \and
    \inferrule{
      \IsNf[\LockCx{\Gamma}]{u}{\Uni}
    }{
      \IsNf{\NfModifyCode{u}}{\Uni}
    }
    \and
    \inferrule{
      \IsNe{e}{\Uni}
      \\
      \IsNe{f}{\Dec{\DecNe{e}}}
    }{
      \IsNf{\NfInj{f}}{\Dec{\DecNe{e}}}
    }
    \and
    \inferrule{
      \IsTm[\LockCx{\Gamma}]{A}{\Uni}<n>
      \\
      \IsNe{e}{\Dec{\ModifyCode{A}}}
    }{
      \IsNe{\NfDecIso{e}}{\Modify{\Dec{A}}}
    }
    \and
    \inferrule{
      \IsNf{u}{\Modify{\Dec{A}}}
    }{
      \IsNf{\NfDecIso*{u}}{\Dec{\ModifyCode{A}}}
    }
  \end{mathparpagebreakable}
  \caption{Definition of selected normals, neutrals, and normal types}
  \label{fig:app:normals-and-neutrals}
\end{figure*}

To ensure that normal forms are $\eta$-long, neutrals can only be `injected' into normals by
$\NfInj{-}$ for types without an $\eta$ law \eg{}, at modal types but not at dependent
products. Finally, we emphasize that normal forms are freely generated so their equality is
decidable if equality of modalities and 2-cells is decidable. This is more subtle than it may
appear at first blush, and we return to this point in Section \ref{sec:normalization:properties}.

\paragraph{Renamings}
While normal and neutral forms are not stable under substitution, they are stable under the
restricted class of \emph{renamings}. The formal definition of renamings is presented in
Figure~\ref{fig:app:renamings}. Intuitively, they are the smallest class of substitutions closed
under weakening, composition, identity, modal substitutions ($\LockSb{-}$,$\Key{\alpha}{}$), and
extension by variables $\Var{k}^\alpha$.

Renamings are easily seen to act on normal forms, neutral forms, and normal types.  Unlike normals
and neutrals, however, renamings are taken up to a definitional equality which ensures that \eg{},
composition is associative and that modal substitutions organize into a 2-functor. This poses no
issue as the action of renamings on normals and neutrals send definitionally equal renamings to
identical normals and neutrals, ensuring that the action lifts to equivalences classes.

A nontrivial definitional equality on renamings is essential, however, as it ensures that the class
of contexts of mode $m$ and renamings between them organizes into a category $\Ren{m}$ and that the
assignments $m \mapsto \Ren{m}$, $\mu \mapsto \LockCx{-}$, and $\alpha \mapsto \Key{\alpha}{}$
define a 2-functor $\Mor{\Coop{\Mode}}{\CAT}$.
\begin{lem}
  The decoding of renamings to substitutions gives a 2-natural transformation
  $\Mor[\EmbRen[-]]{\Ren{-}}{\Cx{-}}$.
\end{lem}

\section{Models and cosmoi}
\label{sec:models}

Gratzer et al.~\cite{gratzer:mtt-journal:2021} introduced \MTT{} as a generalized algebraic theory
so that \MTT{} is automatically equipped with a category of models. A standard result of GATs
ensures that the syntax of \MTT{} organizes into an initial model which opens the possibility
of semantic methods for proving results about syntax. Gratzer et al.~\cite{gratzer:mtt-journal:2021}
then repackages the definition of models in the language of natural models~\cite{awodey:2018}.

\subsection{Natural models of \MTT{}}
\label{sec:cosmoi:models}

We begin by recalling the presentation of a model of \MTT{} given by Gratzer et
al.~\cite{gratzer:mtt-journal:2021}. Recall that a natural model of type theory~\cite{awodey:2018}
is a pair of a category $\CC$---representing a category of contexts---together with a representable
natural transformation $\Mor[\El{}]{\EL{}}{\TY{}}$:
\begin{defi}
  A natural transformation $\Mor[f]{X}{Y} : \PSH{\CC}$ is \emph{representable} when each fiber of
  $f$ over a representable point of $Y$ is itself representable \ie{}, $\Yo{C} \times_Y X$ is
  representable for each $\Mor{\Yo{C}}{Y}$.
\end{defi}

Intuitively, $\El{}$ displays pairs of terms with their types over types. These two objects organize
into presheaves through substitution on terms and types. With this in mind, the representability
condition encodes context extension.

In order to adapt this to \MTT{}, we can no longer consider just a category of contexts. The
existence of multiple modes mandates that we consider a 2-functor of contexts
$\Mor[F]{\Coop{\Mode}}{\CAT}$. The action of modalities $\Mor[F\prn{\mu}]{F\prn{m}}{F\prn{n}}$ gives
the semantic equivalent of $\LockCx{-}$, while the 2-cell component $F\prn{\alpha}$ interprets
$\Key{\alpha}{}$.

Each mode $m : \Mode$ is equipped with a morphism $\Mor[\El{m}]{\EL{m}}{\TY{m}} : \PSH{F\prn{m}}$
representing the terms and types of mode $m$ and each modality $\Mor[\mu]{n}{m}$ induces a functor
which acts by precomposition $\Pre{F\prn{\mu}}$.

\begin{defi}
  A model of \MTT{} without any type constructors is a strict 2-functor
  $\Mor[F]{\Coop{\Mode}}{\CAT}$ together with a collection of morphisms
  $\Mor[\El{m}]{\EL{m}}{\TY{m}} : \PSH{F\prn{m}}$ such that $\Pre*{F\prn{\mu}}{\El{n}}$ is
  representable for each $\Mor[\mu]{n}{m}$.
\end{defi}

Connectives are individually specified on top of this structure. For instance, the following
pullback square in $\PSH{F\prn{m}}$ for each mode $m$ ensures closure under dependent sums:
\begin{equation}
  \DiagramSquare{
    nw/style = pullback,
    width = 6cm,
    height = 1.5cm,
    nw = \Sum{A : \TY{m}} \Sum{B : \ReIdx{\El{m}}{A} \to \TY{m}} \Sum{a : \ReIdx{\El{m}}{A}} \ReIdx{\El{m}}{B\prn{a}},
    sw = \Sum{A : \TY{m}} \Prod{\_ : \ReIdx{\El{m}}{A}} \TY{m},
    ne = \EL{m},
    se = \TY{m},
  }
  \label{diag:cosmoi:dependent-sum}
\end{equation}

Diagram~\ref{diag:cosmoi:dependent-sum} takes advantage of the model of extensional \MLTT{} in a
presheaf topos~\cite{hofmann:1997} and we have written $\ReIdx{\El{m}}{A}$ to denote the
specialization of $\El{m}$ (viewed as a dependent type over $\TY{M}$) with $A$. We will freely take
advantage of this model and use our assumption of a hierarchy of Grothendieck universes to equip it
with an infinite hierarchy of cumulative universes~\cite{hofmann-streicher:1997}. We refer to a
family of presheaves as \emph{small} if it is classified by a universe.

Dependent products $\Fn{A}{B}$ are specified by a similar pullback square but their encoding in
\MTT{} presents a slight complication. Recall that dependent products include a modality
$\Fn{A}{B}$. In order to account for $\mu$, we use $\Pre{F\prn{\mu}}$; if elements of $\TY{m}\prn{X}$
represent types from mode $m$ in context $X : F\prn{m}$, elements
$\Pre{F\prn{\mu}}\prn{\TY{n}}\prn{X}$ represent types from mode $n$ but in context
$F\prn{\mu}\prn{X}$. Accordingly, the presence of dependent products is encoded by the following
pullback square:

\begin{equation}
  \DiagramSquare{
    nw/style = pullback,
    width = 7cm,
    height = 1.75cm,
    nw = {
      \Sum{A : \Pre{F\prn{\mu}}\prn{\TY{n}}}
      \Sum{B : \ReIdx{\Pre{F\prn{\mu}}\prn{\El{n}}}{A} \to \TY{m}}
      \Prod{a : \ReIdx{\Pre{F\prn{\mu}}\prn{\El{n}}}{A}} \ReIdx{\El{m}}{B\prn{a}}
    },
    sw = \Sum{A : \Pre{F\prn{\mu}}\prn{\TY{n}}} \ReIdx{\Pre{F\prn{\mu}}\prn{\El{n}}}{A} \to \TY{m},
    ne = \EL{m},
    se = \TY{m},
  }
  \label{diag:cosmoi:dependent-prod}
\end{equation}

Given $\Mor[\mu]{n}{m}$, we can specify the formation and introduction rules of $\Modify$ with
another commuting square:
\begin{equation}
  \DiagramSquare{
    width = 3.5cm,
    height = 1.5cm,
    nw = \Pre{F\prn{\mu}}{\EL{n}},
    sw = \Pre{F\prn{\mu}}{\TY{n}},
    ne = \EL{m},
    se = \TY{m},
  }
  \label{diag:cosmoi:modal-intro}
\end{equation}
Unlike dependent sums or products, modal types do not have a universal property---an $\eta$ law---so
they cannot be encoded by a single pullback. Instead we must describe the elimination principle
separately. Following Gratzer et al.~\cite{gratzer:mtt-journal:2021}, we encode the elimination
principle as an internal lifting structure.

\begin{defiC}[Definition 18~\cite{awodey:2018}]
  An internal lifting structure $s : i \pitchfork \tau$ between a pair of morphisms $\Mor[i]{A}{B}$
  and $\Mor[\tau]{X}{Y}$ is a section of canonical map $\Mor{X^B}{Y^B \times_{Y^A} X^A}$.
\end{defiC}

Fix a pair of modalities $\Mor[\mu]{n}{m}$ and $\Mor[\nu]{o}{n}$ and write $c$ for the comparison
map $\Mor{\Pre{F\prn{\nu}}(\EL{o})}{\Pre{F\prn{\nu}}\prn{\TY{o}} \times_{\TY{n}} \EL{n}}$ induced
by Diagram~\ref{diag:cosmoi:modal-intro}. The elimination principle for $\nu$-modal types with a
framing modality $\mu$ is encoded by a lifting structure of the following type:
\[
  \Pre{F\prn{\mu}}(c) \pitchfork \Pre{F\prn{\mu\circ\nu}}\prn{\TY{o}} \times \El{m}
  : \SLICE{\PSH{F\prn{o}}}{\Pre{F\prn{\mu\circ\nu}}\prn{\TY{o}}}
\]

This definition is somewhat obstruse, but we will soon be in a position to formulate a far more
intuitive version of it by taking advantage of a richer version of the internal language in
Section~\ref{sec:cosmoi:presheaf-cosmoi}.

As models of a particular GAT, models of \MTT{} assemble into a category. A morphism between models
$F$ and $G$ is given by a 2-natural transformation $\Mor{F}{G}$ along with natural assignments of
terms and types of $F$ to the terms and types of $G$. All of these operations are required to
strictly preserve term, type, and context formers. We refer the reader to
Gratzer et al.~\cite{gratzer:mtt-journal:2021} for a precise description.

Finally, a standard result of GATs is that the \emph{syntactic model} occupies a distinguished place
in the category of models:
\begin{thm}
  \label{thm:cosmoi:initiality}
  Syntax is the initial model of \MTT{}.
\end{thm}

\subsection{\MTT{} cosmoi}
\label{sec:cosmoi:cosmoi}

As mentioned in Section~\ref{sec:introduction}, normalization is proven through the construction of a model of
\MTT{} together with a map from this model to syntax. Models of \MTT{} and morphisms
between them are difficult to construct, however, because of the extreme strictness of morphisms and
the requirement that each $\El{m}$ be a representable natural transformation. Prior to
normalization, therefore, we introduce a weakened notion of model: an \MTT{} cosmos. An \MTT{}
cosmos is an axiomatization of a natural model of \MTT{}, but rather than working in presheaf topoi
and requiring that $\El{m}$ is a representable natural transformation a cosmos requires only that
$\El{m}$ be a morphism in a locally cartesian closed category equipped with structure such as
Diagrams \ref{diag:cosmoi:dependent-prod} and \ref{diag:cosmoi:modal-intro}.

\begin{defi}
  A \emph{cosmos} is a pseudofunctor $\Mor[F]{\Mode}{\CAT}$ such that each $F\prn{m}$ is a locally
  cartesian closed category and each $F\prn{\mu}$ has a left adjoint
  $F_!\prn{\mu} \Adjoint F\prn{\mu}$.
\end{defi}

One should imagine a cosmos $F$ as arising from some model of \MTT{} $F_0$ with
$F\prn{m} = \PSH{F_0\prn{m}}$. The adjunction $F\prn{\mu}_! \Adjoint F\prn{\mu}$ is then recording the
adjunction given by precomposition and left Kan extension $F_0\prn{\mu}_! \Adjoint
F_0\prn{\mu}^*$. In particular, the left adjoint to $F\prn{\mu}$ allows us to capture the left
adjoint action of a modality on contexts ($\LockCx{-}$) while $F\prn{\mu}$ is more intended to
record the modality itself. While this example is strictly 2-functorial, we allow a general cosmos to be
pseudofunctorial. The formal connection between models and cosmoi is given by the following example:

\begin{exa}
  \label{ex:cosmoi:model-to-cosmos}
  A model of \MTT{} $F$ assembles into a cosmos $G$ by taking $G\prn{m} = \PSH{F\prn{m}}$ and
  $G\prn{\mu} = \Pre{F\prn{\mu}}$. In particular, we write $\Mor[\InterpSyn]{\Mode}{\CAT}$ for
  the cosmos induced by the initial model of \MTT{} specified by Theorem~\ref{thm:cosmoi:initiality}.
\end{exa}

The additional requirements imposed by natural models of \MTT{} to encode various connectives can be
transferred \emph{mutatis mutandis} to a cosmos; they are all stated within the language of locally
cartesian closed categories.

\begin{defi}
  An cosmos $F$ is an \MTT{} cosmos when equipped with the following structure:
  \begin{enumerate}
  \item In $F(m)$, there is a universe $\Mor[\El{m}]{\EL{m}}{\TY{m}}$ with a choice of codes
    witnessing its closure under dependent sums and products, identity types, and booleans. For
    instance, a choice of pullback square of the following shape:
    \begin{equation*}
      \DiagramSquare{
        nw/style = pullback,
        width = 7cm,
        height = 2.5cm,
        nw = {
          \Sum{A : F\prn{\mu}\prn{\TY{m}}}
          \Sum{B : \ReIdx{F\prn{\mu}\prn{{\El{n}}}}{A} \to \TY{m}}
          \Prod{a : \ReIdx{F\prn{\mu}\prn{{\El{n}}}}{A}} \ReIdx{\El{m}}{B\prn{a}}
        },
        sw = \Sum{A : F\prn{\mu}\prn{\TY{n}}} \ReIdx{F\prn{\mu}\prn{{\El{n}}}}{A} \to \TY{m},
        ne = \EL{m},
        se = \TY{m},
        south = \mathbf{Prod},
        north = \mathbf{lam},
      }
    \end{equation*}
  \item For each $\mu$, there exists a chosen commuting square
    \begin{equation}
      \DiagramSquare{
        nw = F(\mu)(\EL{n}),
        sw = F(\mu)(\TY{n}),
        ne = \EL{m},
        se = \TY{m},
        south = \mathbf{Mod},
        width = 4cm,
        height = 1.75cm,
      }
      \label{diag:cosmoi:cosmoi-modal-intro}
    \end{equation}
    \label{point:cosmoi:quasi-representation-1}
  \item For each $\Mor[\mu]{n}{m}$ and $\Mor[\nu]{o}{n}$, there is a chosen lifting structure
    $F(\mu)(m) \pitchfork F(\mu\circ\nu)(\TY{o}) \times \El{m}$, where
    $\Mor[m]{F(\nu)(\EL{o})}{F(\nu)(\TY{o}) \times_{\TY{n}} \EL{n}}$ is the comparison map induced
    by Diagram~\ref{diag:cosmoi:cosmoi-modal-intro}.
    \label{point:cosmoi:quasi-representation-2}
  \item $\El{m}$ contains a subuniverse also closed under all these connectives.
  \end{enumerate}
\end{defi}

\begin{defi}
  \label{def:cosmoi:morphism}
  A morphism between \MTT{} cosmoi $\Mor[\alpha]{F}{G}$ is a 2-natural transformation $\alpha$ such
  that $\alpha_m$ is an LCCC functor and preserves all connectives strictly.

  Furthermore, we require that $\alpha$ satisfies the Beck-Chevalley condition so that there is a
  natural isomorphism $\beta_\mu : \alpha_n \circ F(\mu)_! \cong G(\mu)_! \circ \alpha_m$ commuting
  with transposition. Precisely, if $\Mor[a]{X}{F(\mu)(Y)} : F(m)$ the transposition of
  $\alpha_\mu \circ \alpha_m(a)$ is $\alpha_n(\Transpose{a}) \circ \beta_\mu^{-1}$.
\end{defi}

Definition~\ref{def:cosmoi:morphism} uses a number of concepts from 2-category theory and we take
a moment to recall and discuss them here. First, a 2-natural transformation $\alpha$ between
pseudofunctors $\Mor[F,G]{\Mode}{\CAT}$ consists of a collection of functors
$\Mor[\alpha_m]{F\prn{m}}{G\prn{m}}$ along with a family of natural isomorphisms $\alpha_{\mu}$
witnessing the commutativity of the following diagrams up to natural isomorphism:
\[
  \DiagramSquare{
    nw = F\prn{n},
    sw = F\prn{m},
    ne = G\prn{n},
    se = G\prn{m},
    west = F\prn{\mu},
    east = G\prn{\mu},
    north = \alpha_n,
    south = \alpha_m,
    width = 4cm,
  }
\]
The collection of natural isomorphisms $\alpha_\mu$ satisfy a number of coherence conditions forcing
them to behave as expected with respect to composition and identity in $\Mode$ as well as to force
them to be natural with respect to 2-cells in $\Mode$. Fortunately, these higher conditions will not
generally factor into what follows, so we refer the reader to Johnson and
Yau~\cite{johnson-yau:2020} where this notion is detailed under the name \emph{strong
  transformation}.

Note that $F\prn{m}$ and $G\prn{m}$ are both LCC and equipped with universes closed under various
connectives. The next part of Definition~\ref{def:cosmoi:morphism} requires that $\alpha_{\mu}$
respects this additional structure. Finally, since $F\prn{\mu}$ and $G\prn{\mu}$ are both right
adjoints, one can ask whether there is a natural isomorphism witnessing
$\alpha_m \circ F_!\prn{\mu} = G_!\prn{\mu} \circ \alpha_n$. The final requirement---that
$\alpha_\mu$ satisfy the Beck-Chevalley condition---essentially states that there is such a
natural isomorphism and that it is canonically induced from $\alpha_\mu$. In particular, this
ensures that transposing a morphism along $F_!\prn{\mu} \Adjoint F\prn{\mu}$ and then applying
$\alpha_m$ produces the same result as applying $\alpha_n$ and transposing along
$G_!\prn{\mu} \Adjoint G\prn{\mu}$.

A morphism of \MTT{} cosmoi is both more and less restrictive than a morphism of \MTT{}
models. While a morphism of models need not induce an LCC functor between the relevant presheaf
categories, a morphism of cosmoi is not required to strictly preserve context extension or the
choice of terminal context. It so happens that the only map of consequence in this paper is locally
cartesian closed, so the additional structure of morphisms of cosmoi poses no issue. Not requiring
the strict preservation of context extension and dropping the representability requirements from
\MTT{} cosmoi, however, ensures that cosmoi are far easier to construct.

Merely defining a normalization cosmos $\InterpGl$ and projection
$\Mor[\pi]{\InterpGl}{\InterpSyn}$, however, is not enough to prove normalization; we also need a
section to $\pi$. In the category of models, this section would exist as a consequence of
initiality, but $\InterpSyn$ is not initial in the category of \MTT{} cosmoi.%
\footnote{%
  2-monad theory~\cite{kinoshita:1999,gratzer:lcccs:2020} yields an initial cosmos $\ICosmos$ but
  we work with $\InterpSyn$ because---unlike $\ICosmos$---it is known to adequately represent
  syntax.
}
Accordingly, we cannot easily obtain a section of a map into $\InterpSyn$ and in fact sections
rarely exist. Any such map, however, is essentially surjective on definable terms \eg{}, for any
syntactic context $\Gamma$ there exists some object in $X : G\prn{m}$ along with
$\alpha : \pi\prn{X} \cong \Yo{\Gamma}$. Similar statements hold for terms, types, \etc{} While these
choices need not assemble into a morphism of cosmoi, such piecemeal liftings suffice for the
normalization algorithm in Section~\ref{sec:normalization}.

\begin{thm}
  \label{thm:cosmoi:quasi-initiality}
  Fix an \MTT{} cosmos $G$ and $\Mor[\pi]{G}{\InterpSyn}$.
  \begin{enumerate}
  \item For $\IsCx{\Gamma}$, there exists $\Interp{\Gamma} : G(m)$ and
    a canonical isomorphism $\alpha_\Gamma : \Yo{\Gamma} \cong \pi(\Interp{\Gamma})$.
  \item For every $\IsTy{A}$, there exists $\Mor[\Interp{A}]{\Interp{\Gamma}}{\TY{m}}$ such that
    $\pi(\Interp{A}) \circ \alpha_\Gamma = \YoEm{A}$.
  \item For every $\IsTm{M}{A}$, there exists $\Mor[\Interp{M}]{\Interp{\Gamma}}{\EL{m}}$ lying over
    $\Interp{A}$ such that $\pi(\Interp{M}) \circ \alpha_\Gamma = \YoEm{M}$.
  \end{enumerate}
  Here $\YoEm{-}$ is the isomorphism induced by the Yoneda lemma. Moreover, each lift $\Interp{-}$
  respects definitional equality.
\end{thm}

\begin{rem}
  While we have proven this result quite generally, we will apply it only in the special case
  where $\pi$ is a 2-natural transformation between strict 2-functors and required isomorphisms of
  left adjoints are likewise identities. The reader may accordingly safely ignore these coherences
  when reading the proof without consequence.
\end{rem}

\begin{rem}
  Both Theorem~\ref{thm:cosmoi:initiality} and \ref{thm:cosmoi:quasi-initiality} are categorical
  abstractions of \emph{rule induction}. Indeed, \ref{thm:cosmoi:initiality} is used to prove
  \ref{thm:cosmoi:quasi-initiality}---via the construction of an appropriate displayed
  model~\cite{kaposi:qiits:2019}---and the latter takes the place of rule induction in the proof of
  normalization (see Theorem~\ref{thm:normalization:normalization}).
\end{rem}

\begin{proof}
  We write $\SEl$, $\STy$ and $\STm$ instead of $\El{m}$, $\TY{m}$, and $\EL{m}$ in the syntactic
  model, reserving the latter exclusively for $G$. We write $\Interp{\mu}$ for the functor sending
  $\Gamma$ to $\LockCx{\Gamma}$. We begin by replacing $G$ by an equivalent strict 2-functor so that
  $\pi$ becomes strictly 2-natural.

  We construct a displayed model of \MTT{}~\cite{kaposi:qiits:2019} which lies over the syntactic
  model. Using the existing coherence result for \MTT{}~\cite{gratzer:tech-report:2020}, we only
  ensure that $\LockCx{\LockCx{\Gamma}}<\nu>$ and $\LockCx{\Gamma}<\mu\circ\nu>$ agree up to
  pseudonatural isomorphism.
  \begin{itemize}
  \item A context in $m$ is a triple $X : G(m)$, $\IsCx{\Gamma}$, and
    $\alpha : \pi(X) \cong \Yo{\Gamma}$.
  \item A type in a context $(X, \Gamma, \alpha)$ is a pair of $\Mor[\bar{A}]{X}{\TY{m}}$ and
    $\IsTy{A}$ such that $\pi(\bar{A}) = \YoEm{A} \circ \alpha$.
  \item A term in a context $(X, \Gamma, \alpha)$ of type $(\bar{A}, A)$ is a pair
    $\Mor[\bar{M}]{X}{\El{m}\brk{\bar{A}}}$ and $\IsTm{M}{A}$ such that $\pi(\bar{M}) = \YoEm{M} \circ \alpha$.
  \item A substitution $\Mor{(X, \Gamma, \alpha)}{(Y, \Delta, \beta)}$ is a pair $\Mor[f]{X}{Y}$ and
    $\IsSb{\delta}{\Delta}$ satisfying $\beta \circ \pi(f) = \Yo{\delta} \circ \alpha$
  \end{itemize}
  Once this model is constructed, the result follows from
  Theorem~\ref{thm:cosmoi:initiality}. The construction of contexts, substitutions, terms, and types
  is straightforward as $\pi$ is a 2-natural transformation which preserves finite limits, and
  commutes with all connectives. We show two cases.

  \paragraph{The action of a modality on a context}
  Given a triple $(X, \Gamma, \alpha)$ at mode $m$ and a modality $\Mor[\mu]{n}{m}$, we define the
  `locked' context to be the following:
  \[
    (\LKan{G(\mu)}(X), \LockCx{\Gamma}, \gamma \circ \LKan{\Interp{\mu}}{\alpha} \circ \beta)
  \]
  Here $\beta : \pi(\LKan{G(\mu)} X) \cong \LKan{\Interp{\mu}} \pi(X)$ and
  $\gamma : \LKan{\Interp{\mu}} \Yo{\Gamma} \cong \Yo{\LockCx{\Gamma}}$ are the canonical
  isomorphisms.

  \paragraph{Modal types}
  Suppose we are given a context $(X, \Gamma, \alpha)$ and a type $(\bar{A}, A)$ in the context
  $(\LKan{G(\mu)}(\mu)(X), \LockCx{\Gamma}, \gamma \circ \Interp{\mu}^*(\alpha) \circ
  \beta_\mu)$. Writing $\bar{B}$ for the transpose of $\bar{A}$, we form the modal type as
  \[
    (\mathbf{Mod}_\mu(\bar{B}), \Modify{A})
  \]
  It remains to check that these types are coherent \ie{}:
  \[
    \pi(\CModify({\bar{B}})) = \YoEm{\Modify{A}} \circ \alpha
  \]
  By assumption, $\pi(\bar{B}) = \YoEm{A} \circ \gamma \circ \Interp{\mu}^*(\alpha) \circ \beta$. By
  our assumption that $\pi$ satisfies Beck-Chevalley
  $\pi(\bar{B}) = \Transpose{\YoEm{A} \circ \gamma} \circ \alpha$. The result follows from
  the fact that $\pi$ preserves $\CModify$.
\end{proof}

\subsection{Presheaf cosmoi}
\label{sec:cosmoi:presheaf-cosmoi}

Example~\ref{ex:cosmoi:model-to-cosmos} shows that each model of \MTT{} induces an \MTT{} cosmos. In
fact, such cosmoi are particularly well-behaved as they are comprised of presheaf topoi connected by
adjoint triples. These cosmoi enjoy a privileged role in our proof and we observe some of their
unique behavior.

\begin{defi}
  A presheaf cosmos $F$ is a cosmos where $F$ is a strict 2-functor, each $F\prn{m}$ is a presheaf
  topos, and each right adjoint $F\prn{\mu}$ sends small families to small families.
\end{defi}

What distinguishes presheaf cosmoi from other cosmoi is the rich internal language they
offer. Gratzer et al.~\cite{gratzer:mtt-journal:2021} have proven that such a cosmos $F$ supports a
model of \emph{extensional} \MTT{} with the same mode theory where $\Modify{-}$ is interpreted by
$F\prn{\mu}$. We will now use extensional \MTT{} as a \emph{multimodal metalanguage} to specify the
structure of an \MTT{} cosmos as a sequence of constants, thereby reducing its construction to a
series of programming exercises. It is this characterization of \MTT{}-cosmoi that we will
use in Section~\ref{sec:normalization-cosmos} to construct the normalization cosmos.

\begin{rem}
  Some caution is required here, as a presheaf cosmos will frequently host more than one
  interpretation of \MTT{}, with different universes of types. In particular, if we consider the
  collection of presheaf categories $E = \PSH{F\prn{-}}$ where $F$ is a strict 2-functor coming from
  a model of \MTT{}, we may interpret \MTT{} into $E$ either by choosing types to be arbitrary
  families of presheaves, or locally representable families of presheaves. This is comparable to
  Diagram~\ref{diag:cosmoi:dependent-sum}, where type theory is used to describe a model of type
  theory.
\end{rem}

Within this internal language, the universe $\Mor[\El{m}]{\EL{m}}{\TY{m}}$ is encoded by a pair of
types:
\[
  \Ty{m} : \Uni[0]
  \qquad
  \Tm{m} : \prn{A : \Ty{m}} \to \Uni[0]
\]

Each of the diagrams discussed in Sections~\ref{sec:cosmoi:models} and \ref{sec:cosmoi:cosmoi} can
then be translated into constants within this language with the use of dependent types automatically
encoding commutativity. For instance, Diagram~\ref{diag:cosmoi:cosmoi-modal-intro} becomes the
following pair of constants:
\[
  \ModConst : \DeclNameless{\Ty{n}} \to \Ty{m}
  \qquad
  \ModIntroConst :
  \DeclVar{A}{\Ty{n}}\DeclNameless{\Tm{n}\prn{A}} \to \Tm{m}\prn{\ModConst\prn{A}}
\]

In this language it is far easier to specify the modal elimination principle:
\begin{align*}
  &\ModElimConst :{}\\
  &\quad \DeclVar{A}{\Ty{n}}<\nu\circ\mu>\,\prn{B : \DeclNameless{\Tm{n}\prn{\ModConst\prn{A}}}<\nu> \to \Ty{o}}\\
  &\quad \prn{
    b : \DeclVar{x}{\DelimMin{1}\Tm{n}(A)}<\nu\circ\mu> \to
    \Tm{o}\prn{\DelimMin{1}B\prn{\ModIntroConst\prn{A,x}}}
  }\\
  &\quad \to \DeclVar{a}{\Tm{m}(\ModConst(A))}<\nu>  \to \Tm{o}(B(a))
\end{align*}

Each argument to $\ModElimConst$ corresponds directly to a premise of the rule given in
Section~\ref{sec:mtt}. The hypothetical judgment is encoded by the dependent products in the
language and each occurrence of $\LockCx{-}<->$ is replaced with an occurrence of the corresponding
modal type within the metalanguage. The $\beta$-rule for this elimination principle is encoded by
another constant inhabiting the equality type:
\begin{align*}
  &\ModElimConstEq :{}\\
  &\quad \DeclVar{A}{\Ty{n}}<\nu\circ\mu>\,\prn{B : \DeclNameless{\Tm{n}\prn{\ModConst\prn{A}}}<\nu> \to \Ty{o}}\\
  &\quad \prn{
    b : \DeclVar{x}{\DelimMin{1}\Tm{n}(A)}<\nu\circ\mu> \to
    \Tm{o}\prn{\DelimMin{1}B\prn{\ModIntroConst\prn{A,x}}}
  }\\
  &\quad \to \DeclVar{a}{\Tm{m}(A)}<\nu\circ \mu> \to\ModElimConst\prn{A,B,b,\ModIntroConst\prn{A,a}} = b\prn{a}
\end{align*}

The remaining connectives are detailed in Figure~\ref{fig:cosmoi:internal-constants}.
\begin{figure}
  \small
  \begin{align*}
    &\PiConst : \DeclVar{A}{\Ty{m}}\,\prn{B : \DeclNameless{\Tm{m}\prn{A}} \to \Ty{m}} \to \Ty{m}
    \\
    &\alpha_{\PiConst} : \DeclVar{A}{\Ty{m}}\,\prn{B : \DeclNameless{\Tm{m}\prn{A}} \to \Ty{m}}\\
    &\quad\to \Tm{m}\prn{\PiConst\prn{A,B}} \cong \brk{\DeclVar{a}{\Tm{m}\prn{A}} \to \Tm{m}\prn{B\prn{a}}}
    \\[0.2cm]
    &\SigConst : \prn{A : \Ty{m}} \to \prn{\Tm{m}(A) \to \Ty{m}} \to \Ty{m}
    \\
    &\alpha_{\SigConst} : (A : \Ty{m})(B : \Tm{m}(A) \to \Ty{m})\\
    &\quad\to \Tm{m}(\SigConst(A, B)) \cong \brk{\Sum{a : \Tm{m}(A)} \Tm{m}(B(a))}
    \\[0.2cm]
    &\BoolConst : \Ty{m}
    \\
    &\TrueConst,\FalseConst : \Tm{m}(\BoolConst)
    \\
    &\IfConst : (A : \Tm{m}(\BoolConst) \to \Ty{m})\\
    &\quad \to \Tm{m}(A(\TrueConst)) \to \Tm{m}(A(\FalseConst)) \to (b : \Tm{m}(\BoolConst)) \to \Tm{m}(A(b))
    \\
    \_ &: (A : \Tm{m}(\BoolConst) \to \Ty{m})
    \,(t : \Tm{m}(A(\TrueConst)))
    \,(f : \Tm{m}(A(\FalseConst)))\\
    &\quad{}\to (\IfConst(A, t, f, \TrueConst) = t) \times (\IfConst(A, t, f, \FalseConst) = f)
    \\[0.2cm]
    &\IdConst : (A : \Ty{m})(a_0,a_1 : \Tm{m}(A)) \to \Ty{m}
    \\
    &\ReflConst : (A : \Ty{m})(a : \Tm{m}(A)) \to \Tm{m}(\IdConst(A,a,a))
    \\
    &\IdElimConst : (A : \Ty{m})
    \,(B : (a_0,a_1 : \Tm{m}(A))(p : \Tm{m}(\IdConst(A,a_0,a_1))) \to \Ty{m})\\
    &\quad \to ((a : \Tm{m}(A)) \to \Tm{m}(B(a,a,\ReflConst(a))))\\
    &\quad \to (a_0,a_1 : \Tm{m}(A))(p : \Tm{m}(\IdConst(A,a_0a_1))) \to \Tm{m}(B(a_0,a_1,p))
    \\
    &\_ : (A : \Ty{m})
    \,(B : (a_0,a_1 : \Tm{m}(A))(p : \Tm{m}(\IdConst(A,a_0,a_1))) \to \Ty{m})\\
    &\quad \to (b : (a : \Tm{m}(A)) \to \Tm{m}(B(a,a,\ReflConst(a))))\\
    &\quad \to (a : \Tm{m}(A)) \to \IdElimConst(A,B,b,a,a,\ReflConst(a)) = b(a)
    \\[0.2cm]
    &\UniConst : \Ty{m}
    \\
    &\DecConst : \Tm{m}(\UniConst) \to \Ty{m}
    \\
    &\SigCodeConst : \prn{A : \Tm{m}\prn{\UniConst}}
      \to \prn{\Tm{m}(\DecConst(A)) \to \Tm{m}(\UniConst)} \to \Tm{m}(\UniConst)
    \\
    &\PiCodeConst : \DeclVar{A}{\Tm{n}(\UniConst)}
     \to \prn{\DeclNameless{\Tm{n}(\DecConst(A))} \to \Ty{m}} \to \Tm{m}(\UniConst)
    \\
    &\BoolCodeConst : \Tm{m}(\UniConst)
    \\
    &\ModCodeConst : \DeclNameless{\Tm{n}(\UniConst)} \to \Tm{m}(\UniConst)
    \\
    &\DecIsoConst_{\SigCodeConst} :
    (A : \Tm{m}(\UniConst))(B : \Tm{m}(\DecConst(A)) \to \Tm{m}(\UniConst))\\
    &\quad \to \Tm{m}(\DecConst(\SigCodeConst(A,B))) \cong \Tm{m}(\SigConst(\DecConst(A), \DecConst \circ B))
    \\
    &\DecIsoConst_{\PiCodeConst} :
    \DeclVar{A}{\Tm{n}(\UniConst)}(B : \DeclNameless{\Tm{n}(\DecConst(A))} \to \Tm{m}(\UniConst))\\
    &\quad \to \Tm{m}(\DecConst(\PiCodeConst(A,B))) \cong \Tm{m}(\PiConst(\DecConst(A), \DecConst \circ B))
    \\
    &\DecIsoConst_{\BoolCodeConst} : \Tm{m}(\DecConst(\BoolCodeConst)) \cong \Tm{m}(\BoolConst)
    \\
    &\DecIsoConst_{\ModCodeConst} :
    \DeclVar{A}{\Tm{m}(\UniConst)} \to \Tm{m}(\DecConst(\ModCodeConst(A))) \cong \Tm{m}(\ModConst(\DecConst(A)))
  \end{align*}
  \caption{Internal presentation of an \MTT{} cosmos}
  \label{fig:cosmoi:internal-constants}
\end{figure}

\section{Multimodal Synthetic Tait computability}
\label{sec:mstc}

In light of Section~\ref{sec:models}, we revise the proof outlined in
Section~\ref{sec:introduction}: instead of constructing a glued \emph{model} of \MTT{}, we will
construct a glued \MTT{} \emph{cosmos}. In fact, we will construct a glued presheaf cosmos, and take
advantage of the internal language discussed in Section~\ref{sec:cosmoi:presheaf-cosmoi} to upgrade
it to an \MTT{} cosmos with a projection onto $\InterpSyn$. Prior to this, however, we must show
that (1) a pair of cosmoi can be glued together and (2) that each mode of the internal language of
the resulting cosmos can be extended with synthetic Tait computability primitives compatible with
the already-present \MTT{} modalities.

\subsection{Synthetic Tait computability}
\label{sec:mstc:stc}

For this subsection, fix two presheaf topoi $\EE$ and $\FF$ along with a continuous functor
$\Mor[\rho]{\EE}{\FF}$.
\begin{defi}
  The \emph{Artin gluing} $\GL{\rho}$ is a category whose objects are triples $\prn{E, F, f}$ of an
  object from $\EE$, an object from $\FF$, and a morphism $\Mor{F}{\rho\prn{E}}$. Morphisms in
  $\GL{\rho}$ are commuting squares:
  \[
    \DiagramSquare{
      height = 1.5cm,
      width = 2.75cm,
      nw = F_0,
      ne = F_1,
      sw = \rho\prn{E_0},
      se = \rho\prn{E_1},
      west = f_0,
      east = f_1,
      north = \alpha,
      south = \rho\prn{\beta},
    }
  \]
  Projection induces functors $\Mor[\pi_0]{\GL{\rho}}{\EE}$ and $\Mor[\pi_1]{\GL{\rho}}{\FF}$.
\end{defi}

\begin{exa}
  Intuitively $\GL{\rho}$ is a category of proof-relevant $\FF$-predicates on $\rho$-elements of
  $\EE$.  To cultivate this intuition, consider $\FF = \SET$ and $\rho = \Hom{\ObjTerm{}}{-}$. An
  object of $\GL{\Hom{\ObjTerm{}}{-}}$ is a triple of $(S,E,f)$ which induces a proof-relevant
  predicate $\Phi(e) = f^{-1}\prn{e}$ on the global points of $E$. Following Tait~\cite{tait:1967}, we
  refer to elements in the image of $f$ as \emph{computable elements}.
  Morphisms are then morphisms of $\EE$ equipped with additional structure ensuring that
  computable elements are sent to computable elements.
\end{exa}

We now reap the first reward from considering proof-\emph{relevant} predicates: $\GL{\rho}$ is
extremely well-behaved.
\begin{thmC}[\cite{sga:4,carboni:1995}]
  \label{thm:mstc:gl-psh}
  $\GL{\rho}$ is a presheaf topos and $\pi_0$ is a logical functor with left and right adjoints.
\end{thmC}

As a presheaf topos, $\GL{\rho}$ enjoys a model of extensional type theory with a strictly
cumulative hierarchy of universes and a universe of propositions $\Prop$. We can use this language
to \emph{synthetically} build logical relations models~\cite{sterling:modules:2021}. In order to
effectively construct such models, however, we must supplement type theory with primitives specific
to $\GL{\rho}$. The most fundamental of these is a proposition:
\begin{defi}
  The \emph{syntactic proposition} $\Syn : \Prop$ is interpreted in $\GL{\rho}$ as the subterminal
  object $\prn{\ObjTerm{\EE}, \ObjInit{\FF}, \ArrInit{}}$.
\end{defi}

Recalling the correspondence between objects of $\GL{\rho}$ and predicates, $\Syn$ is the predicate
on $\ObjTerm{\EE}$ with no computable elements. What makes this proposition useful is its ability to
wipe out the obligation to track computable elements. A morphism $\Mor[f]{\Syn \times A}{B}$ must
contain a morphism $\Mor[\pi_0\prn{f}]{\pi_0\prn{\Syn \times A} \cong \pi_0\prn{A}}{\pi_0\prn{B}}$,
but there are no computable elements of $\Syn \times A$ so $\pi_0\prn{f}$ entirely determines $f$;
there is a bijection $\Hom[\GL{\rho}]{\Syn \times A}{B} \cong
\Hom[\EE]{\pi_0\prn{A}}{\pi_0\prn{B}}$. Internally, hypothesizing $\Syn$ collapses the category to
$\EE$:
\begin{lem}
  \label{lem:mstc:open-subtopos}
  There is an equivalence $\EE \Equiv \SLICE{\GL{\rho}}{\Syn}$.
\end{lem}

In topos-theoretic terms, $\EE$ is an open subtopos of $\GL{\rho}$. As an open subtopos, we can
present $\EE$ internally to $\GL{\rho}$ through a lex idempotent monad
$\Open A = \Syn \to A$~\cite{rijke:2020}. This modality has a strongly disjoint lex idempotent
modality, $\Closed A$~\cite[Section 3.4]{rijke:2020}. While we could work with $\Closed$ entirely
through this characterization, it is helpful to fix a definition:
\begin{equation}
  \DiagramSquare{
    height = 1.3cm,
    width = 2.75cm,
    nw = \Syn \times A,
    ne = A,
    sw = \Syn,
    se = \Closed A,
    se/style = pushout,
  }
  \label{diag:mstc:closed}
\end{equation}
Intuitively, $\Closed A$ is the portion of $A$ with a trivial $\EE$ component. This is even clearer
if one calculates the behavior of $\Closed$ on a closed type $A = \prn{E,F,f}$ as
$\Closed A = \prn{\ObjTerm{},F,\ArrTerm{}}$. Just as hypothesizing $\Syn$ \ie{}, working under $\Open$,
recovers $\EE$ internally to $\GL{\rho}$, working under $\Closed$ recovers $\FF$. Phrased in
topos-theoretic terms, $\FF$ is a \emph{closed} subtopos of $\GL{\rho}$.

The final ingredient we must add to our type theory is the \emph{realignment
  axiom}~\cite{orton:2018,birkedal:2019,sterling:modules:2021}, stating that the following
canonical map has an inverse $\Realign$ for any $B : \Uni$:
\begin{equation}
  \prn{\Sum{A : \Uni} \brk{A \cong B}}
  \to
  \prn{\Sum{A : \Syn \to \Uni} \Prod{z : \Syn} A\prn{z} \cong B}
\end{equation}

Unfolding these conditions yields the following:
\begin{defi}
  \label{def:mstc:realign}
  Fix $B : \Uni$, $A : \Open \Uni$, and $\alpha : \Prod{z : \Syn} A\prn{z} \cong B$. The
  \emph{realignment} $\Realign\prn{B,A,\alpha}$ of $B$ along $\alpha$ is a term of type
  $\Sum{A^* : \Uni} A^* \cong B$ satisfying the following condition:
  \[
    \Prod{z : \Syn} \Realign\prn{B,A,\alpha} = \prn{A\prn{z},\alpha\prn{z}}
  \]
\end{defi}

More intuitively, realignment states that a predicate lying over an object in $\EE$ can be shifted
to lie over an isomorphic object. A proper motivation of realignment is deferred to its use in
Section~\ref{sec:normalization-cosmos}, but broadly realignment will be used to satisfy the strict
equalities demanded by Definition~\ref{def:cosmoi:morphism} where a priori two constants might agree only up to
isomorphism.

Theorem 8.4 of Orton and Pitts~\cite{orton:2018} shows that a Hofmann--Streicher universe satisfies realignment for
levelwise decidable propositions. Using the presentation of $\GL{\rho}$ as a presheaf
topos~\cite{carboni:1995}, $\Syn$ is clearly levelwise decidable and so realignment at $\Syn$ is
constructively valid. Indeed, for this proposition realignment has a simple and intuitive meaning.
To a first approximation, it allows us to take an object in a gluing topos $\Mor{X}{\rho\prn{Y}}$ along
with an isomorphism $Y \cong Y'$ and perturb the first object to $\Mor{X}{\rho\prn{Y'}}$. Making
this precise (\eg{}, allowing $\Realign$ to act in an arbitrary context) is only marginally more
complex.

\begin{defi}
  The language of synthetic Tait computability is extensional type theory with a cumulative
  hierarchy of universes and a universe of propositions equipped with a distinguished proposition
  $\Syn : \Prop$ such that each universe satisfies the realignment axiom for $\Syn$.
\end{defi}

This subsection is summarized by the following result, which might be termed the `fundamental lemma'
of STC:
\begin{thm}
  \label{thm:mstc:stc}
  $\GL{\rho}$ is a model of STC.
\end{thm}

\subsection{Gluing together cosmoi}
\label{sec:mstc:mstc}

While a model in $\GL{\rho}$ for a carefully chosen $\EE$, $\FF$, and $\rho$ is sufficient to prove
many results of \MLTT{}~\cite{coquand:2019} the situation for \MTT{} is more complex. Rather than
gluing along a single functor, it is necessary to glue along an entire 2-natural transformation of
continuous functors between 2-functors of presheaf topoi. We begin by considering a pair of presheaf
cosmoi for the mode theory $\brc{\Mor[\mu]{n}{m}}$ and a 2-natural transformation of
right adjoints between them:
\begin{equation}
  \DiagramSquare{
    height = 1.35cm,
    width = 2.75cm,
    nw = \EE_n,
    sw = \EE_m,
    ne = \FF_n,
    se = \FF_m,
    north = \rho_n,
    south = \rho_m,
    west = f,
    east = g,
  }
  \label{diag:mstc:gluing}
\end{equation}
For simplicity and since we do not require the additional generality, we shall assume that $F$ and
$G$ are strict 2-functors and that the 2-natural transformation between them is likewise strict.
Let us further assume that $f$ and $g$ preserve finite colimits.

Gluing `horizontally', we obtain a pair of categories $\GL{\rho_n}$ and $\GL{\rho_m}$ and by
Theorems \ref{thm:mstc:gl-psh} and \ref{thm:mstc:stc} both are presheaf topoi and models of
STC. Artin gluing is functorial, and Diagram~\ref{diag:mstc:gluing} induce a functor
$\Mor[\GL{f,g}]{\GL{\rho_n}}{\GL{\rho_m}}$ sending $\prn{E_n,F_n,x}$ to
$\prn{f\prn{E_n},g\prn{F_n},g\prn{x}}$.
\begin{lem}
  \label{lem:mstc:gl-is-right-adjoint}
  $\Mor[\GL{f,g}]{\GL{\rho_n}}{\GL{\rho_m}}$ is a right adjoint.
\end{lem}
\begin{proof}
  While this follows classically from the special adjoint functor theorem, an explicit construction
  is useful. There is a comparison $\Mor[\beta]{g_! \circ \rho_m}{\rho_n \circ f_!}$ induced by
  transposition and the unit of the $f_! \Adjoint f$. The left adjoint $\GL{f,g}_!$ sends
  $\Mor[f]{F}{\rho_m\prn{E}}$ to $\Mor[\beta \circ g_!\prn{f}]{g_!\prn{F}}{\rho_n\prn{f_!\prn{E}}}$.
  The isomorphism $\Hom{\brk{f,g}_!\prn{X}}{Y} \cong \Hom{X}{\brk{f,g}\prn{Y}}$ is given
  component-wise by the isomorphisms associated with $f_! \Adjoint f$ and $g_! \Adjoint g$.
\end{proof}

\begin{rem}
  This explicit calculation show that $\Mor[\pi_n]{\GL{\rho_n}}{\EE_{n}}$ and
  $\Mor[\pi_m]{\GL{\rho_m}}{\EE_{m}}$ assemble into a natural transformation which satisfies
  Beck-Chevalley.
\end{rem}

Since each $\GL{\rho_{-}}$ is a presheaf topos, it supports a model of extensional type theory. We
wish to stitch these models together into a single model of \MTT{} with mode theory
$\brc{\Mor{n}{m}}$ using the results of Gratzer et al.~\cite{gratzer:mtt-journal:2021}. To do so, we
must show that $\GL{f,g}$ induces a dependent right adjoint between models of \MLTT{} in
$\GL{\rho_{n}}$ and $\GL{\rho_{m}}$. Next, we show this holds if we take the models of extensional
type theory in $\GL{\rho_{-}}$ as each having universes of types given by a sufficiently large
Hofmann--Streicher universe:
\begin{lem}
  \label{lem:mstc:gl-is-dra}
  The adjunction $\GL{f,g}_! \Adjoint \GL{f,g}$ induces a dependent right adjoint with respect to
  sufficiently large Hofmann--Streicher universe $\UU$.
\end{lem}
\begin{proof}
  It suffices to argue that $\GL{f,g}$ sends a $\UU$-small family in $\GL{\rho_n}$ to a $\UU$-small
  in $\GL{\rho_m}$. This is proven by \eg{}, Gratzer et al.~\cite[Lemma 3.3.7]{gratzer:universes:2022}.
\end{proof}

As a consequence of Lemma~\ref{lem:mstc:gl-is-dra}, we obtain a model of \MTT{} with the mode theory
$\brc{\Mor[\mu]{n}{m}}$ which interprets $n$, $m$, and $\mu$ as $\GL{\rho_n}$, $\GL{\rho_m}$, and
$\GL{f,g}$ respectively. This model of \MTT{} is particularly well-behaved: equality is extensional
and $\GL{f,g}$ validates the strong transposition-style elimination rules specified by
Birkedal et al.~\cite{birkedal:2020}.
\begin{lem}
  \label{lem:mstc:syn}
  In this model of \MTT{}, $\Modify[\mu]{\Syn_n} \cong \Syn_m$
\end{lem}
\begin{proof}
  Externally, $\Syn_n = \prn{\ObjTerm{},\ObjInit{},\ArrInit{}}$ but $g$ preserves $\ObjInit{}$ while
  $f$ preserves $\ObjTerm{}$, so
  $\GL{f,g}\prn{\Syn_n} \cong \prn{\ObjTerm{},\ObjInit{},\ArrInit{}} = \Syn_m$.
\end{proof}

\begin{lem}
  In this model of \MTT{}, $\Open \Modify{A} \cong \Modify{\Open A}$ and
  $\Closed \Modify{A} \cong \Modify{\Closed A}$.
\end{lem}
\begin{proof}
  We consider the only case of $\Open$, as the argument for $\Closed$ is identical. First, we
  observe that $\GL{f,g}$ preserves $\Open$ \emph{externally}. That is, there is an isomorphism
  $\alpha : \GL{f,g} \circ \Open \cong \Open \circ \GL{f,g}$. It remains to show that this
  isomorphism can be internalized. Let us write $\Mor[\El{m}]{\EL{m}}{\TY{m}}$ for the universe of
  types in $\GL{\rho_m}$ and write $\El{n}$ for its counterpart in $\GL{\rho_n}$. Let us further
  write $i$, $\Dep{\Open}_m$, and $\Dep{\Open}_n$ for the cartesian natural transformations
  $\Mor{\GL{f,g}\prn{\El{n}}}{\El{m}}$, $\Mor{\Open\El{m}}{\El{m}}$, and $\Mor{\Open\El{n}}{\El{n}}$
  that are used to interpret $\Modify{-}$ and $\Open$ in both $\GL{\rho_n}$ and $\GL{\rho_m}$,
  respectively.

  Unfolding this statement into the model, we must argue that the following pair of maps classify
  isomorphic families:
  \[
    \begin{tikzpicture}[diagram]
      \node (AA) {$\GL{f,g}\prn{\Open \TY{n}}$};
      \node [right = 5cm of AA] (AB) {$\GL{f,g}\prn{\TY{n}}$};
      \node [right = 5cm of AB] (AC) {$\TY{m}$};
      \node [below = 2cm of AA] (BA) {$\GL{f,g}\prn{\Open \TY{n}}$};
      \node [right = 5cm of BA] (BB) {$\Open \TY{m}$};
      \node [right = 5cm of BB] (BC) {$\TY{m}$};
      \path[->] (AA) edge node[above] {$\GL{f,g}\prn{\Dep{\Open}}$} (AB);
      \path[->] (AB) edge node[above] {$i$} (AC);
      \path[->] (BA) edge node[above] {$\Open i \circ \alpha$} (BB);
      \path[->] (BB) edge node[above] {$\Dep{\Open}$} (BC);
    \end{tikzpicture}
  \]
  We check that both classify $\GL{f,g}\prn{\Open\El{n}}$ as both $\GL{f,g}$ and $\Open$ preserve
  finite limits.
\end{proof}

\begin{rem}
  Technically, $\Syn$, $\Open$, and $\Closed$ should be always annotated with a mode. In light of
  these results, however, we shall omit this annotation and systematically \emph{identify} $\Syn_m$
  and $\Modify{\Syn_n}$. As both are subterminal, there are no coherence issues in this
  identification.
\end{rem}

\begin{defi}
  The language of \emph{multimodal STC} (MSTC) is extensional \MTT{} with a cumulative hierarchy of
  universes and a universe of propositions such that
  \begin{itemize}
  \item Each mode is equipped with a proposition $\Syn$.
  \item Each universe satisfies the realignment axiom for $\Syn$.
  \item \MTT{} modalities commute with $\Syn$, $\Open$, and $\Closed$.
  \end{itemize}
\end{defi}

Summarizing the preceding discussion:
\begin{thm}
  $\GL{\rho_n}$, $\GL{\rho_m}$, and $\GL{f,g}$ assemble into a presheaf cosmos and a model of MSTC.
\end{thm}

In fact, it is only a small step from this result to the full fundamental lemma of multimodal STC:
\begin{thm}
  \label{thm:mstc:fundamental}
  Given a pair of cosmoi $\Mor[F,G]{\Mode}{\CAT}$ and a 2-natural transformation $\Mor[\rho]{F}{G}$
  such that each $F\prn{\mu},G\prn{\mu}$ preserves finite colimits and each $\rho_m$ is continuous,
  $\Mor[\GL{\rho}]{\Mode}{\CAT}$ both a presheaf cosmos and a model of MSTC. Furthermore
  $\Mor[\pi_0]{\GL{\rho}}{F}$ is a morphism of cosmoi.
\end{thm}

\section{The normalization cosmos}
\label{sec:normalization-cosmos}

Recall from Section~\ref{sec:primer:normals-and-neutrals} the 2-functor of categories of renamings
$\Ren{-}$. By an identical construction to Example~\ref{ex:cosmoi:model-to-cosmos}, we obtain the
cosmos of renamings $\InterpRen{-} = \PSH{\Ren{-}}$ and the 2-natural transformation
$\Mor[\EmbRen[-]]{\Ren{-}}{\Cx{-}}$ acts by precomposition to yield a 2-natural transformation
$\Mor[\InvRen[-]]{\InterpSyn}{\InterpRen}$. Theorem~\ref{thm:mstc:fundamental} then yields the
following:
\begin{defi}
  \label{def:model:gluing-cosmos}
  The normalization cosmos $\InterpGl$ is a presheaf cosmos and model of MSTC where
  $\InterpGl{m} = \GL{\InvRen}$.
\end{defi}

\begin{rem}
  One may explicitly present $\GL{\InvRen}$ as a presheaf category over the \emph{collage} of
  $\Ren{m}$ and $\Cx{m}$~\cite{carboni:1995}. This is a category whose objects are given by the
  disjoint union of $\Ren{m} \Coprod{} \Cx{m}$ and with morphisms defined as follows:
  \begin{align*}
    &\Hom{\In{0}\prn{\Delta}}{\In{0}\prn{\Gamma}} = \Hom[\Ren{m}]{\Delta}{\Gamma}
    &
    &\Hom{\In{1}\prn{\Delta}}{\In{1}\prn{\Gamma}} = \Hom[\Cx{m}]{\Delta}{\Gamma}
    \\
    &\Hom{\In{1}\prn{\Delta}}{\In{0}\prn{\Gamma}} = \Hom[\Cx{m}]{\Delta}{i\prn{\Gamma}}
    &
    &\Hom{\In{0}\prn{\Delta}}{\In{1}\prn{\Gamma}} = \emptyset
  \end{align*}
\end{rem}

As a further consequence of Theorem~\ref{thm:mstc:fundamental}, the projection map
$\Mor[\pi_0]{\InterpGl}{\InterpSyn}$ is a morphism of cosmoi. In this section, we equip
$\InterpGl$ with the structure of an \MTT{} cosmos and show that $\pi_0$ extends to a morphism of
\MTT{} cosmoi.

\subsection{Prerequisites for the normalization cosmos}
\label{sec:model:prereq}

Before we extend $\InterpGl$ to an \MTT{} cosmos, we import features of $\InterpGl$ into the
language of MSTC to specialize the latter to this situation. In this section, we begin using the
interpretation of \MTT{} to work internally to $\InterpGl$ and explicitly record the extensions to
MSTC required for the normalization proof.

\begin{nota}[Dependent open modality]
  As $\Open A = \Syn \to A$, we will write $\Open[z] A\prn{z} = \prn{z : \Syn} \to A\prn{z}$ for the
  \emph{dependent} version of the open modality.
\end{nota}

\begin{nota}[Extension types]
  Given a type $A$, a proposition $\phi$, and an element $a : \phi \to A$, we write
  $\Ext{A}{x : \phi}{a\prn{x}}$ for subtype of $A$ of elements equal to $a$ under $\phi$. Formally:
  \[
    \Ext{A}{x : \phi}{a\prn{x}} = \Sum{a' : A}{\prn{x : \phi} \to a' = a\prn{x}}
  \]
  We treat the coercion $\Ext{A}{x : \phi}{a\prn{x}} \to A$ as silent and refer to the equation
  $a' = a\prn{x}$ as a \emph{boundary condition}.
\end{nota}

Recall from Example~\ref{ex:cosmoi:model-to-cosmos} that $\InterpSyn$ already contains the structure of an
\MTT{} cosmos. As a presheaf cosmos, this manifests through a series of constants in the internal
language of $\InterpSyn$. Using Lemma~\ref{lem:mstc:open-subtopos} we import these constants into
$\InterpGl$.

\begin{extension}
  \label{ext:model:syntactic-model}
  For each $m : \Mode$, there is a pair of constants $\IsTm[z : \Syn]{\Ty{m}\prn{z}}{\Uni[0]}$ and
  $\IsTm[z : \Syn, A : \Ty{m}\prn{z}]{\Tm{m}\prn{z,A}}{\Uni[0]}$. These constants are further
  equipped with operations {\`a} la Figure~\ref{fig:cosmoi:internal-constants} closing them under
  dependent sums, dependent products, modal types, \etc{}
\end{extension}

Next, observe that normals, neutrals, and normal types are equipped with an action by renamings, so
that they can be structured as presheaves over $\Ren{-}$. The decoding operations further organize
them into proof-relevant predicates over terms and types \eg{}, the presheaf of normal types as an
object of $\InterpGl$ lying over the presheaf of types from $\InterpSyn{m}$. In fact, because
renamings map variables to variables, the collection of variables of a given type organizes into a
presheaf over $\Ren{-}$ and part of an object in $\InterpGl$. We import these objects into the
internal language as additional constants:
\begin{extension}
  \label{ext:model:nf-ne}
  Given $m : \Mode$ and $A : \Open[z] \Ty{m}\prn{z}$, we have constants
  $\Nf{m}\prn{A},\Ne{m}\prn{A},\Vars{m}\prn{A} : \Ext{\Uni[0]}{z : \Syn}{\Tm{m}\prn{z,A\prn{z}}}$
  and $\NfTy{m} : \Ext{\Uni[0]}{z : \Syn}{\Ty{m}\prn{z}}$.

  We treat the coercion from $\Vars{m}\prn{A}$ to $\Ne{m}\prn{A}$ as silent.
\end{extension}

\begin{nota}
  We frequently omit explicitly passing $z : \Syn$ as an argument to $M : \Open X$. For instance,
  given $A,B : \Open \Ty{m}$ we write $\Nf{m}\prn{\PiConst\prn{A,B}}$ not
  $\Nf{m}\prn{\lambda z.\ \PiConst\prn{z, A\prn{z},B\prn{z}}}$.
\end{nota}

Following Hofmann~\cite{hofmann:1999}, the constructors for normal forms, neutrals, and normal types
can be realized in $\PSH{\Ren{-}}$ by a form of higher-order abstract syntax. As
$\Nf{m}\prn{A}$, $\Ne{m}\prn{A}$, and $\NfTy{m}$ lie over $\Tm{m}\prn{A}$ and $\Ty{m}$, one can
extend this higher-order abstract syntax presentation to $\InterpGl$ and realize each normal form,
neutral, and normal type as a constant of $\Nf{m}\prn{A}$, $\Ne{m}\prn{A}$, or $\NfTy{m}$ which
collapses to the appropriate syntactic constant under $z : \Syn$. As a simple example, the normal
form type for booleans along with the ordinary boolean type former induce maps
$\Mor[\NfBool]{\ObjTerm{}}{\pi_1\prn{\NfTy{m}}}$ and $\Mor[\Bool]{\ObjTerm{}}{\Ty{m}}$ in
$\PSH{\Ren{m}}$ and $\PSH{\Cx{m}}$ respectively. These maps pair together to introduce a morphism
$\Mor[{\Interp{\CBool}}]{\ObjTerm{}}{\Interp{\NfTy{m}}}$ in $\InterpGl{m}$ where we rely on the
equation $\EmbNfTy{\NfBool} = \Bool$ to ensure that these morphisms fit into the commutative square
required by $\InterpGl{m}$. The full collection of constants is specified in
Figure~\ref{fig:model:neutral-and-normal}.

\begin{extension}
  \label{ext:model:nf-ne-const}
  There are constants internalizing normals, neutrals, and normal types.
\end{extension}

\begin{figure}
  \begin{align*}
    &\CPi : \prn{A : \NfTy{m}}(B : \Vars{m}(A) \to {\NfTy{m}}) \to \NfTy{m}
    \\
    &\CSig : (A : \NfTy{m})(B : \Vars{m}(A) \to \NfTy{m}) \to \NfTy{m}
    \\
    &\CId : (A : \NfTy{m}) \to \Nf{m}(A) \to \Nf{m}(A) \to \NfTy{m}
    \\
    &\CBool : \NfTy{m}
    \\
    &\CModify : \DeclNameless{\NfTy{n}} \to \NfTy{m}
    \\[0.2cm]
    &\CLam : \prn{A : \Open \Ty{m}}(B : {\Open\Tm{m}(A)} \to {\Open \Ty{m}})\\
    &\quad \to (\prn{a : \Vars{m}(A)} \to {\Nf{m}(B(a))}) \to \Nf{m}(\PiConst(A,B))
    \\
    &\CApp : \DeclVar{A}{\Open \Ty{m}}(B : {\Open \Tm{m}(A)} \to {\Open \Ty{m}})\\
    &\quad \to \Ne{m}(\PiConst(A,B)) \to \DeclVar{a}{\Nf{m}(A)} \to \Ne{m}(B(a))
    \\[0.2cm]
    &\CNfInj : \Ne{m}(\BoolConst) \to \Nf{m}(\BoolConst)
    \\
    &\CTrue,\CFalse : \Nf{m}(\BoolConst)
    \\
    &\CIf : (A : \Vars{m}(\BoolConst) \to \NfTy{m})\\
    &\quad \to \Nf{m}(A(\TrueConst)) \to \Nf{m}(A(\FalseConst)) \to (b : \Ne{m}(\BoolConst)) \to \Ne{m}(A(b))
    \\[0.2cm]
    &\CNfInj : (A : \Open \Ty{m})(a_0,a_1 : \Open \Tm{m}(A))\\
    &\quad\to \Ne{m}(\IdConst(A,a_0,a_1)) \to \Nf{m}(\IdConst(A,a_0,a_1))
    \\
    &\CRefl{} : (A : \Open[z] \Ty{m}(z))(a : \Open[z] \Tm{m}(z, A(z))) \to \Nf{m}(\IdConst(A,a,a))
    \\
    &\CIdRec : (A : \Open \Ty{m})
      \,(B : (a_0,a_1 : \Vars{m}(A))(p : \Vars{m}(\IdConst(A,a_0,a_1))) \to \NfTy{m})\\
    &\quad \to
      ((a : \Vars{m}(A)) \to \Nf{m}(B(a,a,\ReflConst(a))))\,
      (a_0,a_1 : \Open[z] \Tm{m}(A))(p : \Ne{m}(\IdConst(A,a_0,a_1))) \\
    &\quad \to \Ne{m}(B(a_0,a_1,p))
    \\[0.2cm]
    &\CNfInj : \DeclVar{A}{\Ty{n}} \to \Ne{m}\prn{\ModConst\prn{A}} \to \Nf{m}\prn{\ModConst\prn{A}}
    \\
    &\CMkBox : \DeclVar{A}{\Open \Ty{n}}\DeclNameless{\Nf{n}(A)} \to \Nf{m}\prn{\lambda z.\ \ModConst\prn{z, A(z)}}
    \\
    &\CLetMod : \DeclVar{A}{\Open \Ty{n}}<\nu\circ\mu>
      \, \prn{B : \DeclVar{a}{\Vars{m}\prn{\ModConst\prn{A}}}<\nu> \to \NfTy{o}}\\
    &\quad{} \to \prn{\DeclVar{a}{\Vars{n}(A)}<\nu\circ\mu> \to \Nf{o}\prn{B\prn{\ModIntroConst\prn{a}}}}
      \to \DeclVar{a}{\Ne{m}\prn{\ModConst\prn{A}}}<\nu> \to \Ne{o}\prn{B\prn{a}}
    \\[0.2cm]
    &\CUni : \NfTy{m}
    \\
    &\CDec : \Nf{m}(\UniConst) \to \NfTy{m}
    \\
    &\CNfInj : \Ne{m}(\UniConst) \to \Nf{m}(\UniConst)
    \\
    &\CModifyCode : \DeclNameless{\Nf{n}(\UniConst)} \to \Nf{m}(\UniConst)
    \\
    &\CDecIso_{\CModifyCode} : \DeclVar{A}{\Nf{n}(\UniConst)}
      \to \Nf{m}(\ModConst(A)) \to \Nf{m}(\DecConst(\ModCodeConst(A)))
    \\
    &\CDecIso*_{\CModifyCode} : \DeclVar{A}{\Nf{n}(\UniConst)}
      \to \Ne{m}(\DecConst(\ModCodeConst(A))) \to \Ne{m}(\ModConst(A))
  \end{align*}

  \caption{Neutral and normal forms, internally}
  \label{fig:model:neutral-and-normal}
\end{figure}

Finally, inspecting Definition~\ref{def:model:gluing-cosmos} reveals that modalities are interpreted
by functors which are both left and right adjoints as they preserve all (co)limits. As a result,
modalities preserve coproducts:
\begin{extension}
  \label{ext:model:modalities-cocont}
  $\Modify{A + B} \cong \Modify{A} + \Modify{B}$
\end{extension}

\subsection{The \MTT{} cosmos}
\label{sec:model:mtt-cosmos}

We now extend $\InterpGl$ to an \MTT{} cosmos. To ensure that $\pi_0$ induces a morphism of \MTT{}
cosmoi, it suffices to ensure that each constant we add to $\InterpGl$ is equal to the corresponding
piece of $\InterpSyn$ as internalized by Extension~\ref{ext:model:syntactic-model} under $z : \Syn$.

\paragraph{The universe of computable types and terms}
We begin with the definition of types and terms in this cosmos. Concretely, we require the following
for each $m : \Mode$:
\begin{align*}
  &\Ty*{m} : \Ext{\Uni[2]}{z : \Syn}{\Ty{m}\prn{z}}
  \\
  &\Tm*{m} : \prn{A : \Ty*{m}} \to \Ext{\Uni[1]}{z : \Syn}{\Tm{m}\prn{z,A}}
\end{align*}

We start with the following putative definition of types:
\begin{equation}
  \begin{make-rcd}{T}{\Uni[2]}
    \Code : \NfTy{m}
    \\
    \Pred : \Ext{\Uni[1]}{z : \Syn}{\Tm{m}(z, \Code)}
    \\
    \ReflectFld : \Ext{\Ne{m}\prn{\Code} \to \Pred}{\Syn}{\ArrId{}}
    \\
    \ReifyFld : \Ext{\Pred \to \Nf{m}\prn{\Code}}{\Syn}{\ArrId{}}
  \end{make-rcd}
  \label{cons:model:T}
\end{equation}
In prose, $A : T$ contains the code of a normal type $A.\Code$ as well as a proof-relevant predicate
on the elements of $A.\Code$.

The last two fields ensure that (1) all elements tracked by this predicate can be assigned normal
forms, and (2) all neutrals lie within the predicate. We write $\Reify{A}$ and $\Reflect{A}$ for
$A.\ReifyFld$ and $A.\ReflectFld$. Of the two, the $\ReifyFld$ is the crucial operation needed for
the normalization algorithm: it ensures that computable elements can be given normal
forms. Tait~\cite{tait:1967}, however, has shown that the pair of operations is necessary to close
all type formers under just $\ReifyFld{}$.

We cannot simply define $\Ty*{m} = T$, as $T$ does not satisfy the equation
$z : \Syn \vdash T = \Ty{m}\prn{z}$. It does, however, satisfy this condition up to isomorphism:
under $z : \Syn$, the types of $\Pred{}$, $\ReflectFld{}$, and $\ReifyFld{}$ collapse to singletons,
while the type of $\Code$ collapses to $\Ty{m}\prn{z}$ by Extension~\ref{ext:model:nf-ne}:
\begin{align*}
  &\alpha_{\Open} : \Prod{z : \Syn} T \cong \Ty{m}\prn{z}
  \\
  &\alpha_{\Open}(z,A) = A.\Code
\end{align*}

Observe $\prn{\Ty{m},\alpha_{\Open}} : \Sum{A : \Open\Uni} \Prod{z : \Syn} A\prn{z} \cong T$, so the
realignment axiom of Definition~\ref{def:mstc:realign} applies and we can define
\begin{equation}
  \prn{\Ty*{m},\alpha} = \Realign\prn{T,\Ty{m},\alpha_{\Open}}
  \label{cons:model:ty}
\end{equation}
The equation $z : \Syn \vdash \Ty*{m} = \Ty{m}\prn{z}$ follows immediately from the second half of
Definition~\ref{def:mstc:realign}. On elements $A : \Ty*{m}$, this implies
$z : \Syn \vdash A = \alpha\prn{A}.\Code$. For readability, we continue to use record notation to
manipulate $\Ty*{m}$.

Given $A : \Ty*{m}$, we define $\Tm*{m}\prn{A}$:
\begin{equation}
  \Tm*{m}\prn{A} = A.\Pred
  : \Ext{\Uni[1]}{z : \Syn}{\Tm{m}\prn{z,A}}
  \label{cons:model:tm}
\end{equation}
To see that this is well-typed, we must show $\Tm*{m}\prn{A} = \Tm{m}\prn{z,A}$ given $z :
\Syn$. The type of $A.\Code$ in Construction~\ref{cons:model:T} ensures
$\Tm*{m}\prn{A} = \Tm{m}\prn{z,A.\Code}$. We have observed that $A = A.\Code$ under $z : \Syn$ so
$\Tm*{m}\prn{A} = \Tm{m}\prn{z,A}$.

\paragraph{Type connectives}
It remains only to close $\prn{\Ty*{m},\Tm*{m}}$ under all connectives in such a way that each
connective lies over the corresponding one in $\prn{\Ty{m},\Tm{m}}$. For mode-local connectives,
these constructions are very similar to those given by Sterling~\cite{sterling:phd}
(Lemmas~\ref{lem:model:sums},~\ref{lem:model:booleans},~\ref{lem:model:id},
and~\ref{lem:model:universes}).  Modal types and dependent products, however, involve modalities
and thus are different than the other connectives
(Lemmas~\ref{lem:model:pi} and~\ref{lem:model:mod}).

\begin{lem}
  \label{lem:model:pi}
  $\prn{\Ty*{m},\Tm*{m}}$ is closed under dependent products and the relevant constants lift those
  of $\Ty{m}$ (\ie{}, under an assymption $z : \Syn$, they agree with those of $\Ty{m}$ and
  $\Tm{m}$):
  \begin{align*}
    &\PiConst* : \DeclVar{A}{\Ty*{n}}\prn{B : \DeclNameless{\Tm*{n}\prn{A}} \to \Ty*{m}} \to \Ty*{m}
    \\
    &\alpha_{\PiConst*} : \DeclVar{A}{\Ty*{n}}\prn{B : \DeclNameless{\Tm*{n}\prn{A}} \to \Ty*{m}}\\
    &\quad \to \Tm*{m}\prn{\PiConst*\prn{A,B}} \cong \brk{\DeclVar{a}{\Tm*{n}\prn{A}} \to \Tm*{m}\prn{B\prn{a}}}
  \end{align*}
\end{lem}
\begin{proof}
  We must define two constants ($\PiConst*$ and $\alpha_{\PiConst*}$) with the aforementioned types.
  We begin by fixing $\DeclVar{A}{\Ty*{m}}$ and $B : \DeclNameless{\Tm*{n}\prn{A}} \to
  \Ty*{m}$ and define $\Phi$ as follows:
  \[
    \Phi = \DeclVar{a}{\Tm*{n}\prn{A}} \to \Tm*{m}\prn{B\prn{a}}
  \]
  Observe under $z : \Syn$, the following equality holds:
  \[
    \Phi = \DeclVar{a}{\Tm{n}\prn{z,A}} \to \Tm{m}\prn{B\prn{z,a}}
  \]
  We may apply realignment using
  $\alpha_{\PiConst}\prn{z} : \Tm{m}\prn{z,\PiConst\prn{z,A,B}} \cong \Phi$.  This realignment
  yields a type $\Psi$ and isomorphism $\beta : \Psi \cong \Phi$. Under $z : \Syn$, these restrict
  to $\Tm{m}\prn{z,\PiConst\prn{z,A,B}}$ and $\alpha_{\PiConst}\prn{z}$ respectively.

  With these to hand we define $\PiConst*$ and $\alpha_{\PiConst*}$ as follows:
  \begin{align*}
    &\PiConst*\prn{A,B}.\Code = \CPi\prn{A.\Code, \lambda v.\ B\prn{\Reflect{A} v}.\Code}
    \\
    &\PiConst*\prn{A,B}.\Pred = \Psi
    \\
    &\PiConst*\prn{A,B}.\ReflectFld =
      \lambda e.\ \beta^{-1}\prn{\lambda a.\ \Reflect{B\prn{a}} \CApp\prn{e,\Reify{A} a}}
    \\
    &\PiConst*\prn{A,B}.\ReifyFld =
      \lambda f.\ \CLam\prn{\lambda v.\ \Reify{B\prn{\Reflect{A} v}} \beta\prn{f}\prn{\Reflect{A} v}}
    \\
    &\alpha_{\PiConst*} = \beta
  \end{align*}
  It remains to check a variety of boundary conditions under $z : \Syn$. In particular, we must show
  that $\PiConst*\prn{A,B} = \PiConst\prn{z,A,B}$ and that $\ReflectFld$ and $\ReifyFld$ become the
  identity. These follow directly from assumptions about $A$, $B$, and the boundaries of various
  constructors. For instance
  \begin{align*}
    \PiConst*\prn{A,B}
    &= \PiConst*\prn{A,B}.\Code\\
    &= \CPi\prn{A.\Code, \lambda v.\ B\prn{\Reify{A} v}.\Code}\\
    &= \PiConst\prn{z, A.\Code, \lambda v.\ B\prn{\Reify{A} v}.\Code}\\
    &= \PiConst\prn{z, A, \lambda v.\ B\prn{\Reify{A} v}}\\
    &= \PiConst\prn{z, A, B} \qedhere
  \end{align*}
\end{proof}

\begin{lem}
  \label{lem:model:mod} $\prn{\Ty*{m},\Tm*{m}}$ is closed under modal types and the four relevant
  constants ($\ModConst*$, $\ModIntroConst*$, $\ModElimConst*$, and $\ModElimConstEq*$) lift those of
  their counterparts in $\Ty{m}$ and $\Tm{m}$.
\end{lem}
\begin{proof}
  Fix a modality $\Mor[\mu]{n}{m}$. In this case we define the four constants $\ModConst$,
  $\ModIntroConst$, $\ModElimConst$, and $\ModElimConstEq$ described in
  Section~\ref{sec:cosmoi:models}, subject to the expected boundary conditions. Fix a variable $A :
  \Ty*{n}$ under the modal annotation $\mu$ \ie{}, $\DeclVar{A}{\Ty*{n}}$. We define the unaligned
  predicate as follows:
  \begin{equation*}
    \begin{make-rcd}{\Phi}{\Uni[1]}
      \TmFld : \Nf{m}\prn{\ModConst\prn{A}}
      \\
      \PrfFld :
      \Closed \prn{
        \DelimMin{4}
        \begin{aligned}
          &\Sum{e : \Ne{m}\prn{\ModConst\prn{A}}} \TmFld = \CNfInj\prn{e}
          \\
          + &\Sum{a : \Modify{A.\Pred{}}} \TmFld = \CMkBox\prn{\Reify{A} a}
        \end{aligned}
      }
    \end{make-rcd}
  \end{equation*}
  For the first time, we have used the closed modality $\Closed$ to explicitly tweak the
  proof-relevant predicate. Intuitively, $\Phi$ is a predicate on $\Tm{m}\prn{z,\ModConst\prn{z,A}}$
  and $\TmFld$ ensures that this predicate tracks elements with normals forms. The second field,
  moreover, ensures that these normal are either neutral or $\NfMkBox{a}$ where $a$ is computable.
  Without the closed modality shielding the second field of $\Phi$, however, this could never have
  the correct extent along $z : \Syn$. Using $\Open \Closed X \cong \ObjTerm{}$ and the boundary of
  $\Nf{m}\prn{\ModConst\prn{A}}$, we can now define the following isomorphism:
  \[ \alpha_{\Open}\prn{z,p} = p.\TmFld : \Prod{z : \Syn} \Phi \cong
  \Tm{m}\prn{z,\ModConst\prn{z,A}} \]

  Realigning $\Phi$ along $\alpha_{\Open}$, we obtain $\Psi$ and $\alpha : \Psi \cong \Phi$ which
  under $z : \Syn$ become $\Tm{m}\prn{z,\ModConst\prn{z,A}}$ and $\alpha_{\Open}$.

  We now define $\ModConst*$:
  \begin{align*}
    &\ModConst*\prn{A}.\Code = \CModify\prn{A.\Code}
    \\
    &\ModConst*\prn{A}.\Pred = \Psi
    \\
    &\ModConst*\prn{A}.\ReflectFld = \lambda e.\ \alpha^{-1}\gls{\CNfInj\prn{e}, \eta_\Closed\In{1}\gls{e,\star}}
    \\
    &\ModConst*\prn{A}.\ReifyFld = \lambda m.\ \alpha\prn{m}.\TmFld
  \end{align*}

  Unlike Lemma~\ref{lem:model:pi}, the introduction and elimination principles are not automatically
  obtained from $\alpha$ and they must be constructed separately:
  \[
    \ModIntroConst*\prn{A,a} = \alpha^{-1}\gls{\CMkBox\prn{\Reify{A} a}, \eta_\Closed \In{2}\gls{a,\star}}
  \]

  It remains to define the elimination principle $\ModElimConst*$. This is an involved affair and we
  describe it step-by-step. Begin by fixing $\Mor[\nu]{m}{o}$ along with the following:
  \begin{align*}
    &B : \DeclNameless{\Tm*{m}\prn{\ModConst*\prn{A}}}<\nu> \to \Ty{o}
    \\
    &b : \DeclVar{x}{\Tm*{n}(A)}<\nu\circ\mu> \to \Tm*{o}\prn{B\prn{\ModIntroConst*\prn{A,x}}}
    \\
    &\DeclVar{p}{\Tm*{m}\prn{\ModConst*(A)}}<\nu>
  \end{align*}
  We must construct an element of $\Tm*{o}\prn{B\prn{a}}$. We begin by inspecting $p$. As \MTT{}
  modalities in extensional \MTT{} commute with dependent sums, equality, $\Closed$, and---by
  Extension~\ref{ext:model:modalities-cocont}---with finite coproducts, $p$ can be decomposed into the
  following:
  \begin{align*}
    &\DeclVar{\TmFld}{\Nf{m}\prn{\ModConst\prn{A}}}<\nu>
    \\
    &\PrfFld :
    \Closed \prn{
      \DelimMin{4}
      \begin{aligned}
        &\Sum{e : \Modify[\nu]{\Ne{m}\prn{\ModConst\prn{A}}}} \MkBox[\nu]{\TmFld} = \CNfInj \ZApp{} e
        \\
        + &\Sum{a : \Modify[\nu\circ\mu]{A.\Pred{}}} \MkBox[\nu]{\TmFld} = \prn{\CMkBox \circ \Reify{A}} \ZApp{} a
      \end{aligned}
    }
  \end{align*}

  Recall from Diagram~\ref{diag:mstc:closed} that $\Closed X$ is a pushout of $\Syn$ and $X$. To
  define a map out of $\Closed X$, therefore, it suffices to define a map out of $X$ which is
  constant assuming $z : \Syn$. We conclude by scrutinizing $\PrfFld$:
  \begin{equation*}
    \begin{cases}
      \Reflect{} \CLetMod(A, \lambda v.\ B\prn{\Reflect{} v}.\Code{}, \lambda x.\ \Reify{} b\prn{\Reflect{} x}, e)
      &
      \text{if } \PrfFld = \In{1}\prn{\MkBox[\nu]{e}, \_}
      \\
      b(a)
      &
      \text{if } \PrfFld = \In{2}\prn{\MkBox[\nu]{a}, \_}
    \end{cases}
  \end{equation*}
  Given $z : \Syn$, both branches collapse to $\ModElimConst\prn{z,A,B,b,\TmFld}$ so this yields a
  well-defined map. The boundary conditions follow from routine computations.
\end{proof}

\begin{lem}
  \label{lem:model:sums}
  $\prn{\Ty*{m},\Tm*{m}}$ is closed under dependent sums via:
  \begin{align*}
    \SigConst*\prn{A,B} &: \Ty*{m}\\
    \alpha_{\SigConst*} &: \Tm{m}(\SigConst*(A, B)) \cong \Sum{a : \Tm*{m}(A)} \Tm*{m}(B(a))
  \end{align*}
  Moreover, assuming $z : \Syn$ then $\SigConst* = \SigConst$ and $\alpha_{\SigConst*} = \alpha_{\SigConst}$.
\end{lem}
\begin{proof}
  Fixing $A : \Ty*{m}$ and $B : \Tm*{m}(A) \to \Ty*{m}$. 
  We begin by applying realignment to the following:
  \[
    \prn{\Sum{a : A.\Pred} B\prn{a}.\Pred, \alpha_{\SigConst(z)}}
  \]
  This produces $\Psi : \Uni[1]$ and
  $\alpha_{\SigConst*} : \Psi \cong \Sum{a : A.\Pred} B\prn{a}.\Pred$ such that under the assumption
  $z : \Syn$ the following holds:
  \[
    \Psi = \Tm{m}(\SigConst(z, A, B)) \qquad \alpha_{\SigConst*} = \alpha_{\SigConst}\prn{z}
  \]

  We now define $\SigConst*(A,B)$ as follows:
  \begin{align*}
    &\SigConst*(A,B).\Code = \CSig(A.\Code, \lambda v.\ B.\Code(\Reflect{A} v))
    \\
    &\SigConst*(A,B).\Pred = \Psi
    \\
    &\SigConst*(A,B).\ReflectFld =
    \lambda e.\ \alpha_{\SigConst*}^{-1}
      \gls{
        \Reflect{A}(\CProj{0}(e)),
        \Reflect{B(\Reflect{A}(\CProj{0}(e)))}(\CProj{1}(e))
      }
    \\
    &\SigConst*(A,B).\ReifyFld =
    \lambda p.\ \CPair(\Reify{A}(\alpha_{\SigConst*}p.0), \Reify{B(\alpha_{\SigConst*}p.0)}(\alpha_{\SigConst*}p.1))
  \end{align*}
  The fact that $\Reify{}$ and $\Reflect{}$ lie over the identity follows directly from the $\beta$
  and $\eta$ laws of dependent sums in \MTT{}. We show the calculations for $\Reflect{}$. Fix
  $z : \Syn$:
  \begin{align*}
    \Reflect{\SigConst*(A,B)}(e)
    &= \alpha_{\SigConst*}^{-1} \gls{\Reflect{A}(\CProj{0}(e)), \Reflect{B(\Reflect{A}(\CProj{0}(e)))}(\CProj{1}(e))}\\
    &= \alpha_{\SigConst}^{-1} \gls{\CProj{0}(e), \CProj{1}(e)} \\
    &= \alpha_{\SigConst}^{-1} \gls{\alpha_{\SigConst(A,B)}(e)_0, \alpha_{\SigConst(A,B)}(e)_1} \\
    &= e
  \end{align*}

  The fact that $\SigConst*(A,B).\Code{}$ and $\SigConst*(A,B).\Pred{}$ lie over $\SigConst(A,B)$
  and $\Tm{m}(z, \SigConst(z, \allowbreak A, B))$ follows from their definition and realignment.
\end{proof}

\begin{lem}
  \label{lem:model:booleans}
  $\prn{\Ty*{m},\Tm*{m}}$ is closed under booleans and the relevant constants lie over their
  counterparts in $\prn{\Ty{m},\Tm{m}}$.
\end{lem}
\begin{proof}
  We must implement the following constants:
  \begin{align*}
    &\BoolConst* : \Ext{\Ty*{m}}{z : \Syn}{\BoolConst(z)}
    \\
    &\TrueConst* : \Ext{\Tm*{m}(\BoolConst*)}{z : \Syn}{\TrueConst}
    \\
    &\FalseConst* : \Ext{\Tm*{m}(\BoolConst*)}{z : \Syn}{\FalseConst}
    \\
    &\IfConst* : (A : \Tm*{m}(\BoolConst*) \to \Ty*{m})\\
    &\quad \to \Tm*{m}(A(\TrueConst*))\\
    &\quad \to \Tm*{m}(A(\FalseConst*))\\
    &\quad \to (b : \Tm*{m}(\BoolConst*))\\
    &\quad \to \Ext{\Tm*{m}(A(b))}{z : \Syn}{\IfConst(A, t, f, b)}
    \\
    &\_ : (A : \Tm*{m}(\BoolConst*) \to \Ty*{m})\\
    &\quad \to (t : \Tm*{m}(A(\TrueConst*)))\\
    &\quad \to (f : \Tm*{m}(A(\FalseConst*)))\\
    &\quad \to (\IfConst*(A, t, f, \TrueConst*) = t) \times (\IfConst*(A, t, f, \FalseConst*) = f)
  \end{align*}

  First, we define $\Phi$ by realignment:
  \begin{equation*}
    \begin{make-rcd}{\Phi}{\Ext{\Uni[1]}{z : \Syn}{\Tm{m}\prn{z,\BoolConst}}}
      \TmFld : \Nf{m}\prn{\BoolConst}
      \\
      \PrfFld :
      \Closed \prn{
        \DelimMin{4}
        \begin{aligned}
          &\Sum{e : \Ne{m}\prn{\BoolConst}} \TmFld = \CNfInj\prn{e}
          \\
          + &\Sum{b : \mathbf{2}} \TmFld = \mathsf{rec}_{\mathbf{2}}(b; \NfTrue; \NfFalse))
        \end{aligned}
      }
    \end{make-rcd}
  \end{equation*}
  In the above, we have used $\mathsf{rec}_{\mathbf{2}}$ for the ordinary elimination principle for
  $\mathbf{2}$ in $\InterpGl$. We have opted for the names $\mathbf{2}$ and
  $\mathsf{rec}_{\mathbf{2}}$ in the hopes of avoiding ambiguity with $\BoolConst$, $\CIf$, and
  $\IfConst$.

  We may now define $\BoolConst*$:
  \begin{align*}
    &\BoolConst*.\Code = \CBool
    \\
    &\BoolConst*.\Pred = \Phi
    \\
    &\BoolConst*.\ReflectFld = \lambda e. \gls{\CNfInj(e), \Unit(\In{1}(e, \Ax))}
    \\
    &\BoolConst*.\ReifyFld = \lambda b.\ b.\TmFld
  \end{align*}

  It remains to define the introduction and elimination forms.
  \begin{align*}
    \TrueConst* &= \gls{\CTrue, \Unit(\In{2}(0, \Ax))}\\
    \FalseConst* &= \gls{\CFalse, \Unit(\In{2}(1, \Ax))}
  \end{align*}

  The elimination form is defined by constructing a map out of $\Closed X$,  by taking advantage of
  its definition as a pushout (Diagram \ref{diag:mstc:closed}):
  \begin{align*}
    \IfConst*&(A, t_0, t_1, b = \gls{\TmFld, \PrfFld}) =\\
    &
    \begin{cases}
      \IfConst(z, A, t_0, t_1, b) & \PrfFld = \In{1}(z)
      \\
      \Reify{A(b)} \CIf(\lambda v.\ A(\Reflect{} v).\Code{}, \Reify{} t_0, \Reify{} t_1, e)
      & \PrfFld = \In{2}(\In{1}(e, \_))
      \\
      \mathsf{rec}_{\mathbf{2}}\prn{b_0; t_0; t_1} & \PrfFld = \In{2}(\In{2}(b_0, \_))
    \end{cases}
  \end{align*}
  In this definition, three different incarnations of the elimination rule for booleans are
  used. The first branch deals uses $\IfConst\prn{z, \dots}$ which is the elimination rule from the
  syntactic model, the second uses the neutral form $\CIf$ associated to $\IfConst$, and the third
  is the ``ordinary'' elimination principle for booleans available within the model.
\end{proof}

\begin{lem}
  \label{lem:model:id}
  $\prn{\Ty*{m},\Tm*{m}}$ is closed under intensional identity types and the relevant constants lie
  over their counterparts in $\prn{\Ty{m},\Tm{m}}$.
\end{lem}
\begin{proof}
  We must implement the following constants:
  \begin{align*}
    &\IdConst* : (A : \Ty*{m})(a_0,a_1 : \Tm*{m}(A))\\
    &\quad \to \Ext{\Ty*{m}}{z : \Syn}{\IdConst(z, A, a_0,a_1)}
    \\
    &\ReflConst* : (A : \Ty*{m})(a : \Tm*{m}(A))\\
    &\quad \to \Ext{\Tm*{m}(\IdConst*(A,a,a))}{z : \Syn}{\ReflConst(z, A, a)}
    \\
    &\IdElimConst* : (A : \Ty*{m})\\
    &\quad \to (B : (a_0,a_1 : \Tm*{m}(A)) \to \Tm*{m}(\IdConst*(A,a_0,a_1)) \to \Ty*{m})\\
    &\quad \to (b : (a : \Tm*{m}(A)) \to \Tm*{m}(B(a,a,\ReflConst(a))))\\
    &\quad \to (a_0,a_1 : \Tm*{m}(A))(p : \Tm*{m}(\IdConst*(A,a_0,a_1)))\\
    &\quad \to \Ext{\Tm*{m}(B(a_0,a_1,p))}{z : \Syn}{\IdElimConst(z, B, b, p)}
    \\
    \_ &: (A : \Ty*{m})\\
    &\quad \to (B : (a_0,a_1 : \Tm*{m}(A)) \to \Tm*{m}(\IdConst*(A,a_0,a_1)) \to \Ty*{m})\\
    &\quad \to (b : (a : \Tm*{m}(A)) \to \Tm*{m}(B(a,a,\ReflConst(a))))\\
    &\quad \to (a : \Tm*{m}(A)) \to \IdElimConst*(A,B,b,\ReflConst*(a)) = b(a)
  \end{align*}

  Fix $A : \Ty*{m}$ and $a_0,a_1 : \Tm*{m}\prn{A}$. Just as with the normalization structure for
  booleans, we begin by defining $\Phi$ by realignment:
  \begin{equation*}
    \begin{make-rcd}{\Phi}{\Ext{\Uni[1]}{z : \Syn}{\Tm{m}\prn{z,\IdConst\prn{A,a_0,a_1}}}}
      \TmFld : \Nf{m}\prn{\IdConst\prn{A,a_0,a_1}}
      \\
      \PrfFld :
      \Closed \prn{
        \DelimMin{4}
        \begin{aligned}
          &\Sum{e : \Ne{m}\prn{\IdConst\prn{A,a_0,a_1}}} \TmFld = \CNfInj\prn{e}
          \\
          + &\Sum{a : A.\Pred} a_0 = a_1 \times \TmFld = \ReflConst\prn{\Reify{A} a}
        \end{aligned}
      }
    \end{make-rcd}
  \end{equation*}

  We now define $\IdConst*$:
  \begin{align*}
    &\IdConst*(A,a_0,a_1).\Code{} = \NfId{\Code{A}}{\Reify{A} a_0}{\Reify{A} a_1}
    \\
    &\IdConst*(A,a_0,a_1).\Pred = \Phi
    \\
    &\IdConst*(A,a_0,a_1).\ReflectFld = \lambda e. \gls{\CNfInj(e), \Unit(\In{1}(e, \Ax))}
    \\
    &\IdConst*(A,a_0,a_1).\ReifyFld = \lambda p.\ p.\TmFld
  \end{align*}

  We define reflexivity by $\ReflConst* = \gls{\ReflConst, \Unit(\In{2}(\Ax, \Ax, \Ax))}$. Finally, the
  elimination principle is defined using the induction principle for $\Closed X$.
    \begin{align*}
    &\IdElimConst*(B, b, a_0, a_1, p = \gls{\TmFld, \PrfFld}) =\\
    &
    \begin{cases}
      \IdElimConst(z, B, b, a_0, a_1, p) & \PrfFld = \In{1}(z)
      \\
      \Reify{}
      \CIdRec(
        \lambda l,r,p.B(\Reflect{}l, \Reflect{}r, \Reflect{} p).\Code{},
        \lambda a. \Reify{} b(\Reflect{} a),
        e)
      & \PrfFld = \In{2}(\In{1}(e, \_))
      \\
        b(a_0) & q = \In{2}(\In{2}(\_, \_, \_)) \tag*{\qedhere}
    \end{cases}
  \end{align*}
\end{proof}

\begin{lem}
  \label{lem:model:universes}
  $\prn{\Ty*{m},\Tm*{m}}$ is closed under a universe and the relevant constants lie
  over their counterparts in $\prn{\Ty{m},\Tm{m}}$.
\end{lem}
\begin{proof}
  We begin by constructing the two constants for the universe and the decoding family:
  \begin{align*}
    &\UniConst* : \Ext{\Ty*{m}}{z : \Syn}{\UniConst}
    \\
    &\DecConst* : (A : \Tm*{m}(\UniConst*)) \to \Ext{\Ty*{m}}{z : \Syn}{\DecConst(A)}
  \end{align*}

  At this point we take advantage of the fact that $\Pred{}$ is an element of $\Uni[1]$; in
  particular, we observe that $\Uni[0]$ is small enough to fit inside $\Uni[1]$.

  We may then define $\Psi$ by realigning the following element of $\Uni[1]$ along the evident
  isomorphism to $\Tm*{m}(z,\UniConst(z))$:
  \begin{equation*}
    \begin{make-rcd}{\Psi}{\Ext{\Uni[1]}{z : \Syn}{\Tm*{m}\prn{z,\UniConst}}}
      \Code{} : \Nf{m}(\UniConst)
      \\
      \Pred{} : \Ext{\Uni[0]}{z : \Syn}{\Tm{m}(z, \DecConst(\Code{}))}
      \\
      \ReflectFld{} : \Ext{\Ne{m}(\DecConst(\Code{})) \to \Pred{}}{z : \Syn}{\ArrId{}}
      \\
      \ReifyFld{} : \Ext{\Pred{} \to \Nf{m}(\DecConst(\Code{}))}{z : \Syn}{\ArrId{}}
    \end{make-rcd}
  \end{equation*}
  With $\Psi$ in hand, we may define $\UniConst*$:
  \begin{align*}
    &{\UniConst*}.\Code = \CUni
    \\
    &{\UniConst*}.\Pred = \Psi
    \\
    &{\UniConst*}.\ReflectFld = \lambda e.\ \gls{\CNfInj(e); \Ne{m}\prn{\DecConst\prn{e}}; \ArrId{}; \lambda e.\ \CNfInj(e)}
    \\
    &{\UniConst*}.\ReifyFld = \lambda A.\ A.\Code
  \end{align*}

  The definition of $\DecConst*$ is essentially cumulativity:
  \[
    \DecConst*(\gls{\Code{}; \Pred{}; \ReifyFld{}; \ReflectFld{}}) =
    \gls{\NfDec{\Code{}}; \Pred{}; \ReifyFld{}; \ReflectFld{}}
  \]

  It remains to show that $\prn{\UniConst*,\DecConst*}$ is closed under various type formers. We
  show a representative cases: modal types. This concretely entails implementing the following
  constants:
  \begin{align*}
    &\ModCodeConst* :
    \DeclVar{A}{\Tm*{n}\prn{\UniConst*}} \to
    \Ext{
      \Tm*{m}\prn{\UniConst*}
    }{
      z : \Syn
    }{
      \ModCodeConst\prn{z,A}
      }
    \\
    &\DecIsoConst*_{\ModCodeConst} : \DeclVar{A}{\Tm*{n}(\UniConst*)}\\
    &\quad
    \to
    \Ext{
      \Tm*{m}\prn{\DecConst*\prn{\ModCodeConst*\prn{A}}}
      \cong \Tm*{m}\prn{\ModConst*\prn{\DecConst*\prn{A}}}
    }{
      z : \Syn
    }{
      \DecIsoConst_{\ModCodeConst}\prn{z,A}
    }
  \end{align*}

  Fix $\DeclVar{A}{\Tm*{n}\prn{\UniConst*}}$. We realign
  $\Tm*{m}\prn{\ModConst*\prn{\DecConst*\prn{A}}}$ along the isomorphism
  $\DecIsoConst_{\ModCodeConst}$ to obtain a type $\Psi$ and an isomorphism:
  \[
    \DecIsoConst*_{\ModConst} :
    \Ext{
      \Psi %
      \cong \Tm*{m}\prn{\ModConst*\prn{\DecConst*\prn{A}}}
    }{
      z : \Syn
    }{
      \DecIsoConst_{\ModCodeConst}\prn{z,A}
    }
  \]

  It remains only to define $\ModCodeConst*\prn{A}$ such that $\ModCodeConst*\prn{A}.\Pred = \Psi$:
  \begin{align*}
    &\ModCodeConst*\prn{A}.\Code = \ModifyCode{A.\Code}
    \\
    &\ModCodeConst*(A).\Pred{} = \Psi
    \\
    &\ModCodeConst*(A).\ReflectFld{} =
    \lambda e.\ (\DecIsoConst*_{\ModCodeConst})^{-1}\prn{\Reflect{\ModConst*\prn{\DecConst*\prn{A}}} \NfDecIso{e}}
    \\
    &\ModCodeConst*(A).\ReifyFld{} =
    \lambda m.\ \NfDecIso*{\Reify{\ModConst*\prn{\DecConst*\prn{A}}} \DecIsoConst*_{\ModCodeConst}(m)}
  \end{align*}
  The checks that all constructions lie over their syntactic counterparts follow immediately from
  the conclusions of realignment.
\end{proof}

\begin{thm}
  \label{thm:model:mtt-cosmos}
  $\InterpGl$ supports an \MTT{} cosmos built around $\prn{\Ty*{m},\Tm*{m}}$ and
  $\Mor[\pi_0]{\InterpGl}{\InterpSyn}$ is a map of \MTT{} cosmoi.
\end{thm}

\section{The normalization algorithm}
\label{sec:normalization}

After Theorem~\ref{thm:model:mtt-cosmos}, it remains only to parlay the existence of the
normalization cosmos into a normalization function.

\subsection{The normalization function}
At this point, it becomes necessary to shift from working purely internally to $\InterpGl$ to
inspecting some constructions externally. Accordingly, we will have use for the \emph{total} spaces
of terms and normal forms \eg{} $\Tm*{m} = \Sum{A : \Ty*{m}} \Tm*{m}\prn{A}$. We write $\TY{m}$ and
$\EL{m}$ for the presheaves of types and terms in $\InterpSyn{m}$ to disambiguate them from
$\Ty*{m}$ and $\Tm*{m}$.
\begin{lem}
  \label{lem:normalization:reify}
  There is a morphism $\Mor[\Reify{}]{\Tm*{m}}{\Nf{m}}$ which restricts to $\ArrId{}$ under $\Syn$.
\end{lem}
\begin{proof}
  Working internally, $\Reify{} \prn{A,M} = \prn{A, \Reify{A} M}$.
\end{proof}

Fix a term $\IsTm{M}{A}$. Theorems \ref{thm:cosmoi:quasi-initiality} and \ref{thm:model:mtt-cosmos}
define a map $\Mor[\Interp{M}]{\Interp{\Gamma}}{\Tm*{m}}$ in $\InterpGl{m}$ along with an
isomorphism $\alpha : \pi_0\prn{\Interp{\Gamma}} \cong \Yo{\Gamma}$ such that
$\pi_0\prn{\Interp{M}} = \YoEm{M} \circ \alpha$.

We would like to obtain a normal form for $M$ from $\Interp{M}$. To this end, we can unfold
$\Interp{M}$ along with $\Reify{}$ from Lemma~\ref{lem:normalization:reify} to obtain a commuting
diagram:
\begin{equation*}
  \begin{tikzpicture}[diagram,baseline=(current bounding box.center)]
    \SpliceDiagramSquare{
      height = 1.6cm,
      width = 3.75cm,
      nw = \pi_1\prn{\Interp{\Gamma}},
      sw = \InvRen{\Yo{\Gamma}},
      ne = \pi_1\prn{\Tm*{m}},
      se = \InvRen{\EL{m}},
      south = {\InvRen{\YoEm{M}}},
      west = {\InvRen{\alpha} \circ \Interp{\Gamma}},
    }
    \node [right = 2.5cm of ne] (nf) {$\pi_1\prn{\Nf{m}}$};
    \path[->] (ne) edge (nf);
    \path[->] (nf) edge (se);
  \end{tikzpicture}
\end{equation*}

To normalize $M$, it suffices to construct $\Atoms{\Gamma} : \pi_1\prn{\Interp{\Gamma}}_\Gamma$ such
that $\alpha\prn{\Interp{\Gamma}\prn{\Atoms{\Gamma}}} = \ArrId{} : \InvRen{\Yo{\Gamma}}_\Gamma$:
pushing $\Atoms{\Gamma}$ along the top of the diagram would yield a normal form (an element of
$\pi_1\prn{\Nf{m}}$) which decodes to $M$ by Yoneda. Modulo technical details, $\Atoms{\Gamma}$ is
produced by using $\Reflect{}$ to convert variables for each element of $\Gamma$ into elements of
$\pi_1\prn{\Interp{\Gamma}}$.

\begin{lem}
  \label{lem:normalization:atoms}
  For any $\IsCx{\Gamma}$ there exists
  $\Mor[\Atoms{\Gamma}]{\prn{\Yo{\Gamma},\Yo{\Gamma}}}{\Interp{\Gamma}}$ in $\InterpGl$ lying
  over $\ArrId{} : \InvRen{\Yo{\Gamma}}$ in $\InterpSyn$.
\end{lem}
\begin{proof}
  This proof proceeds by induction on $\Gamma$.
  \begin{description}
  \item[Case] $\Gamma = \EmpCx$\\
    Here $\Interp{\Gamma}$ is terminal, so $\Atoms{\EmpCx}$ is its unique element. The requirement
    that $\Atoms{\EmpCx}$ lie over $\ArrId{}$ is then tautological.
  \item[Case] $\Gamma = \ECx{\Delta}{A}$\\
    In this case, we note that
    $\Interp{\Gamma} =
    \Interp{\Delta} \times_{\InterpGl{\mu}\prn{\Ty*{n}}} \InterpGl{\mu}\prn{\Tm*{n}}$
    and, since pullback are computed pointwise, it suffices to
    construct element of $\pi_1\prn{\Interp{\Delta}_\Gamma}$ and
    $\pi_1\prn{\InterpGl{\mu}\prn{\Tm*{n}}_{\Gamma}}$ separately which agree on
    $\pi_1\prn{\InterpGl{\mu}\prn{\Ty*{n}}_{\Gamma}}$.

    First, we reindex $\Atoms{\Delta}$ by $\IsSb{\Wk}{\Delta}$ to
    obtain $\delta \in \pi_1\prn{\Interp{\Delta}_\Gamma}$. Next, using the element
    $\Var{0} \in \InterpGl{\mu}\prn{\Ne{n}\prn{A}}_{\Gamma}$. It is easily seen that these agree
    on $\pi_1\prn{\InterpGl{\mu}{\prn{\Ty*{n}}_{\Gamma}}}$. The check that this lies over $\ArrId{}$
    follows from the fact that (1) $\delta$ lies over $\Wk$, (2) $\Reflect{A} \Var{0}$ lies over
    $\Var{0}$ and (3) that $\ESb{\Wk}{\Var{0}} = \ArrId{}$.
  \item[Case] $\Gamma = \LockCx{\Delta}$\\
    We define $\Atoms{\Gamma} = \InterpGl{\mu}_!\prn{\Atoms{\Delta}}$.
    The check that this lies over $\ArrId{}$ amounts to the equation in syntax that
    $\LockCx{\ISb{}} = \ISb$.
    \qedhere
  \end{description}
\end{proof}

\begin{rem}
  $\Atoms{\Gamma}$ is analogous to the initial environment used in classical NbE proofs to kick off
  normalization. Abel~\cite{abel:2013}, for instance, denotes the environment $\uparrow^{\Gamma}$.
\end{rem}

Combining Lemma~\ref{lem:normalization:atoms} with the argument above, we conclude that for term
$\IsTm{M}{A}$, there exists $\IsNf[\Gamma]{u}{A}$ such that $\EmbNf{u} = M$. Moreover, because we
have consistently worked with equivalences class of terms, this function automatically respects
definitional equality. Summarizing:

\begin{thm}
  \label{thm:normalization:normalization}
  There is a function $\Normalize{-}{A}$ sending terms of type $\IsTy{A}$ to normal forms such that
  \begin{enumerate}
  \item If $\IsTm{M}{A}$ then $\EqTm{\EmbNf{\Normalize{M}{A}}}{M}{A}$.
  \item If $\EqTm{M}{N}{A}$ then $\Normalize{M}{A} = \Normalize{N}{A}$.
  \end{enumerate}
\end{thm}

We can repeat this process to normalize types instead of terms. Given $\IsTy{A}$, we obtain
$\Mor[\Interp{A}]{\Interp{\Gamma}}{\Ty*{m}}$ which unfolds to an analogous diagram with only a small
change: rather than using $\Reflect{}$ to pass from $\pi_1\prn{\Tm*{m}}$ to normal forms, we use
$\Code{}$ to shift from $\Ty*{m}$ to normal types:
\begin{equation*}
  \begin{tikzpicture}[diagram]
    \SpliceDiagramSquare{
      height = 1.6cm,
      width = 3.75cm,
      nw = \pi_1\prn{\Interp{\Gamma}},
      sw = \InvRen{\Yo{\Gamma}},
      ne = \pi_1\prn{\Ty*{m}},
      se = \InvRen{\TY{m}},
      south = {\InvRen{\YoEm{A}}},
      west = {\alpha \circ \Interp{\Gamma}},
    }
    \node [right = 2.5cm of ne] (nf) {$\pi_1\prn{\NfTy{m}}$};
    \path[->] (ne) edge (nf);
    \path[->] (nf) edge (se);
  \end{tikzpicture}
\end{equation*}

By again pushing $\Atoms{\Gamma}$ along the top of this diagram, we obtain a normalization function
for types.
\begin{thm}
  \label{thm:normalization:ty-normalization}
  There is a function $\NormalizeTy{-}$ sending types to normal types such that
  \begin{enumerate}
  \item If $\IsTy{A}$ then $\EqTy{\EmbNfTy{\NormalizeTy{A}}}{A}$.
  \item If $\EqTy{A}{B}$ then $\NormalizeTy{A} = \NormalizeTy{B}$.
  \end{enumerate}
\end{thm}

\subsection{Corollaries of normalization}
\label{sec:normalization:properties}

A number of important theorems follow as corollaries of
Theorems~\ref{thm:normalization:normalization} and \ref{thm:normalization:ty-normalization}. For
instance, we can reduce the decidability of conversion to the decidability of normal forms.
\begin{cor}
  \label{cor:normalization:dec}
  \leavevmode
  \begin{enumerate}
  \item $\EqTm{M}{N}{A}$ iff $\Normalize{M}{A} = \Normalize{N}{A}$.
  \item $\EqTy{A}{B}$ iff $\NormalizeTy{A} = \NormalizeTy{B}$.
  \end{enumerate}
\end{cor}
\begin{proof}
  We show only the proof for this first claim. The `only if' direction is established by the second
  point of Theorem~\ref{thm:normalization:normalization}. Suppose instead
  $\Normalize{M}{A} = \Normalize{N}{A}$, so
  $\EmbNf{\Normalize{M}{A}} = \EmbNf{\Normalize{N}{A}}$. By the first point of
  Theorem~\ref{thm:normalization:normalization}, $\EmbNf{\Normalize{M}{A}} = M$ and
  $\EmbNf{\Normalize{M}{A}} = N$, so the conclusion follows.
\end{proof}

A priori, however, a given term could have multiple normal forms which complicates further
analysis. We therefore strengthen Theorem~\ref{thm:normalization:normalization} with the following:
\begin{thm}[Tightness]
  \label{thm:normalization:idempotent}
  \leavevmode
  \begin{enumerate}
  \item If $\IsNf[\Gamma]{u}{A}$, then $\Normalize{\EmbNf{u}}{A} = u$.
  \item If $\IsNfTy[\Gamma]{\tau}$, then $\NormalizeTy{\EmbNfTy{\tau}} = \tau$.
  \end{enumerate}
\end{thm}
\begin{proof}
  Recall that Theorems \ref{thm:cosmoi:quasi-initiality} and \ref{thm:model:mtt-cosmos} induce a
  function $\Interp{-}$ sending a piece of syntax to its interpretation in the normalization
  model. Furthermore, recall the $\Gamma$-element $\Atoms{\Gamma} : \Interp{\Gamma}$
  constructed in Lemma~\ref{lem:normalization:atoms}.

  We begin by strengthening the statement to make it more amenable to induction:
  \begin{enumerate}
  \item If $\IsNe{e}{A}$, then
    $\Interp{\EmbNe{M}}(\Atoms{\Gamma}) = \Reflect{\Interp{A}(\Atoms{\Gamma})} e$
  \item If $\IsNf{u}{A}$, then
    $\Reify{\Interp{A}(\Atoms{\Gamma})} \Interp{\EmbNf{u}}(\Atoms{\Gamma}) = u$.
  \item If $\IsNfTy{\tau}$, then ${\Interp{\EmbNfTy{A}}.\Code(\Atoms{\Gamma})} = \tau$.
  \end{enumerate}
  Here we have identified a code $u$ (resp. $e$) as an $\Gamma$-element of $\Nf{A}$
  (resp. $\Ne{A}$). All three follow straightforwardly from mutual induction and the relevant
  definitions. For instance, if we consider $\IsNfTy{\NfFn{\tau}{\sigma}}$, we calculate as
  follows:
  \begin{align*}
    &\Interp{\EmbNfTy{\NfFn{\tau}{\sigma}}}.\Code(\Atoms{\Gamma})
    \\
    &= {\Interp{\Fn{\EmbNfTy{\tau}}{\EmbNfTy{\sigma}}}.\Code(\Atoms{\Gamma})}
    \\
    &=
    \NfFn{
      {\Interp{\EmbNfTy{\tau}}}.\Code\prn{\Atoms{\Gamma}}
    }{
      {\Interp{\EmbNfTy{\sigma}}}.\Code(\Wk^*\Atoms{\Gamma}, \Reflect{} \NeVar{0}{})
    }
    \\
    &=
    \NfFn{
      {\Interp{\EmbNfTy{\tau}}}.\Code\prn{\Atoms{\Gamma}}
    }{
      {\Interp{\EmbNfTy{\sigma}}}.\Code(\Atoms{\ECx{\Gamma}{A}})
    }
    \\
    &= \NfFn{\tau}{\sigma}
  \end{align*}
  In order to carry out this calculation, we took advantage of not only the definition of dependent
  products in the gluing model, but also the interpretation of HOAS and $\Atoms{}$.
\end{proof}

\begin{cor}
  \label{cor:normalization:normalization-is-iso}
  Normalization is an isomorphism between equivalence classes of terms (resp. types) and normal
  forms (resp. normal types).
\end{cor}
\begin{proof}
  Corollary~\ref{cor:normalization:dec} already shows that normalization is injective and
  Theorem~\ref{thm:normalization:idempotent} provides a section.
\end{proof}

These results imply the injectivity of type constructors, an essential property for implementation.
\begin{cor}
  \label{cor:normalization:pi-inj}
  If $\EqTy{A_0 \to B_0}{A_1 \to B_1}$ then $\EqTy{A_0}{A_1}$ and
  $\EqTy[\ECx{\Gamma}{A_0}<\ArrId{}>]{B_0}{B_1}$.
\end{cor}
\begin{proof}
  Set $\tau_i = \NormalizeTy{A_i}$ and $\sigma_i =
  \NormalizeTy[\ECx{\Gamma}{A_0}<\ArrId{}>]{B_i}$. Unfolding definitions shows that
  $\EmbNfTy{\NfFn{\tau_i}{\sigma_i}} = \EmbNfTy{\tau_i} \to \EmbNfTy{\sigma_i} = A_i \to B_i$.  By
  Corollary~\ref{cor:normalization:normalization-is-iso},
  $\NormalizeTy{A_i \to B_i} = \NfFn{\tau_i}{\sigma_i}$.

  Next, we recall that $\EqTy{A_0 \to B_0}{A_1 \to B_1}$ by assumption, so
  $\NfFn{\tau_0}{\sigma_0} = \NfFn{\tau_1}{\sigma_1}$. As an operation on normal forms, however,
  $\NfFn{-}{-}$ is clearly injective, so $\tau_0 = \tau_1$ and $\sigma_0 = \sigma_1$. The result
  now follows from Corollary~\ref{cor:normalization:dec}.
\end{proof}

In light of Corollary~\ref{cor:normalization:dec}, to decide the equality of terms and types, it suffices to
argue that one may decide the equality of neutral and normal forms along with normal types. For this
purpose, we adapt the bidirectional algorithm given by Altenkirch and
Kaposi~\cite{altenkirch:2017}. This argument goes through essentially without alteration, except
that since certain constructors are annotated with 1- and 2-cells from $\Mode$, we require a
decision procedure for objects in the mode theory. Note that this procedure uses \eg{},
Corollary \ref{cor:normalization:pi-inj}, which is why we have delayed its statement till now.
\begin{cor}
  \label{cor:normalization:type-checking}
  If $\Mode$ is decidable, type checking is decidable.
\end{cor}

Finally, Gratzer et al.~\cite{gratzer:2020} show canonicity for \MTT{} extended with the equality
$\LockCx{\EmpCx} = \EmpCx$. Normalization provides a (heavy-handed) proof of canonicity without this
equation by scrutinizing the definition of normal forms:
\begin{cor}
  \label{cor:normalization:canonicity}
  If $\IsTm[\LockCx{\EmpCx}]{M}{\Bool}$ then $M \in \brc{\True,\False}$.
\end{cor}

\section{Extending \MTT{} with crisp identity induction}
\label{sec:crisp}

To demonstrate the flexibility of the normalization argument given in
Sections~\ref{sec:normalization-cosmos} and \ref{sec:normalization}, we now show how it may be
extended to accommodate modal principles not included in \MTT{}.

Recall that, intuitively, a modality in \MTT{} corresponds to a right adjoint. This intuition is
supported by the fact that \MTT{} modalities commute with products. In an extensional version of
\MTT{}, modalities also commute with (extensional) equality. That is, the following canonical map is
an equivalence:
\begin{equation}
  \DeclVar{x,y}{A} \to \Id{\Modify{A}}{\MkBox{x}}{\MkBox{y}} \to \Modify[\mu]{\Id{A}{x}{y}}
  \label{eq:crisp:modal-ext}
\end{equation}

\begin{rem}
  Constructing this map is slightly intricate. We begin by generalizing:
  \[
    \prn{x,y : \Modify{A}} \to \Id{\Modify{A}}{x}{y}
    \to
    \LetMod{x}[{x'}]{
      \LetMod{y}[{y'}]{
        \Modify[\mu]{\Id{A}{x'}{y'}}
      }
    }
  \]
  In this form, we may use ordinary identity induction followed by modal induction to reduce to 
  $\prn{x : \Modify{A}} \to
    \LetMod{x}[{x'}]{
      \LetMod{x}[{y'}]{
        \Modify[\mu]{\Id{A}{x'}{y'}}
      }
    }$
  and then $\DeclVar{x}{A} \to \Modify[\mu]{\Id{A}{x}{x}}$
  respectively.
\end{rem}

In \emph{intensional} \MTT{}, the same principle is not derivable.
\begin{thm}
  \label{thm:crisp:no-go}
  There exists a model of intensional \MTT{} with one mode $m$ and one modality $\Mor[\mu]{m}{m}$ in
  which Equation~\ref{eq:crisp:modal-ext} is not invertible.
\end{thm}
\begin{proof}
  Consider intensional \MTT{} and define an interpretation of \MTT{} into intensional \MLTT{} which
  interprets both modes as \MLTT{} and sends all non-modal types to their counterparts within
  \MLTT{} and interprets modal connectives as follows:
  \begin{align*}
    &\Interp{\LockCx{\Gamma}} = \Interp{\Gamma}.\Nat
    \\
    &\Interp{\ECx{\Gamma}{A}} = \Interp{\Gamma}.\prn{\Nat \to \Interp{A}}
    \\
    &\Interp{\LockCx{\Gamma}<\ArrId{}>} = \Interp{\Gamma}
    \\
    &\Interp{\ECx{\Gamma}{A}<\ArrId{}>} = \Interp{\Gamma}.\Interp{A}
    \\
    &\Interp{\Modify[\mu]{A}} = \Nat \to \Interp{A}
    \\
    &\Interp{\Modify[\ArrId{}]{A}} = \Interp{A}
    \\
    &\Interp{\MkBox{M}} = \Lam{\Interp{M}}
    \\
    &\Interp{\MkBox[\ArrId{}]{M}} = \Interp{M}
    \\
    &\Interp{\LetMod{M}{N}<\xi>[\chi]} = \Sb{\Interp{N}}{\ESb{\ISb}{\Interp{M}}}
  \end{align*}

  Unfolding the interpretation of Equation~\ref{eq:crisp:modal-ext}, we observe that an inverse to
  this map corresponds to function extensionality for functions $\Nat \to A$. As function
  extensionality is independent of \MLTT{}, there must be no inverse to
  Equation~\ref{eq:crisp:modal-ext} definable within \MTT{}.
\end{proof}

In light of Theorem~\ref{thm:crisp:no-go}, we refer to the existence of an inverse to
Equation~\ref{eq:crisp:modal-ext} as \emph{modal extensionality}. Modal extensionality is useful in
practice. In incarnations of guarded recursion within \MTT{}, for instance, some version of modal
extensionality is required to prove any equalities involving guarded
types~\cite{gratzer:mtt-journal:2021,gratzer:lob:2022}. It is therefore worth investigating whether
modal extensionality is compatible with both normalization and canonicity.\footnote{Like function
  extensionality, it is straightforward to maintain \emph{either} normalization or canonicity in the
  presence of modal extensionality. Ensuring for both simultaneously is far more difficult.}

In work by Shulman~\cite{shulman:2018} and Gratzer~\cite{gratzer:mtt-journal:2021}, \emph{crisp}
induction principles are a variation of the induction principles for types such as $\Bool$ or
$\Id{A}{a_0}{a_1}$ which allow the scrutinee of the induction to occur beneath a modality. Crisp
induction principles are derivable in \MTT{} if the modality has an internal right
adjoint~\cite{gratzer:mtt-journal:2021}, but they are justified in other situations. In particular,
crisp induction for identity types is validated if and only if modal extensionality holds. In
contrast to modal extensionality, however, it is straightforward to directly adapt the proofs of
normalization and canonicity to account for crisp identity induction principles:
\begin{mathpar}
  \inferrule{
    \IsTy[\ECx{\ECx{\ECx{\Gamma}{A}<\mu>}{\Sb{A}{\Wk}}<\mu>}{\Id{\Sb{A}{\Wk[2]}}{\Var{1}}{\Var{0}}}<\mu>]{B}
    \\
    \IsTm[\ECx{\Gamma}{A}<\mu>]{M}{\Sb{B}{\ESb{\ESb{\ESb{\Wk}{\Var{0}}}{\Var{0}}}{\Refl{\Var{0}}}}}
    \\
    \IsTm[\LockCx{\Gamma}]{N_0,N_1}{A}<n>
    \\
    \IsTm[\LockCx{\Gamma}]{P}{\Id{A}{N_0}{N_1}}<n>
  }{
    \IsTm{\CrispIdRec{B}{M}{P}}{\Sb{B}{\ESb{\ESb{\ESb{\ISb}{N_0}}{N_1}}{P}}}
  }
  \\
  \CrispIdRec{B}{M}{\Refl{N}} = \Sb{M}{\ESb{\ISb}{N}}
\end{mathpar}

The modularity of our proof of normalization ensures that only local changes to the construction of
identity types in $\InterpGl$ are needed to adapt the entire proof to support crisp
induction. Concretely, two changes to primitive constants added to MSTC by
Section~\ref{sec:model:prereq}. One alteration to the definition of cosmoi and one to the definition
of neutral forms:
\begin{align*}
  &\IdElimConst_\mu : \DeclVar{A}{\Ty{n}}
  \,(B : \DeclVar{a_0,a_1}{\Tm{n}(A)}\DeclVar{p}{\Tm{n}(\IdConst(A,a_0,a_1))} \to \Ty{m})\\
  &\quad \to (\DeclVar{a}{\Tm{n}(A)} \to \Tm{m}(B(a,a,\ReflConst(a)))) \\
  &\quad \to \DeclVar{a_0,a_1}{\Tm{n}(A)}\DeclVar{p}{\Tm{n}(\IdConst(A,a_0,a_1))}\\
  &\quad \to \Tm{m}(B(a_0,a_1,p))
  \\
  &\CIdRec_\mu : \DeclVar{A}{\Open \Ty{n}}
  \,(B : \DeclVar{a_0,a_1}{\Vars{n}(A)}\DeclVar{p}{\Vars{m}(\IdConst(A,a_0,a_1))} \to \NfTy{m})\\
  &\quad \to (\DeclVar{a}{\Vars{n}(A)} \to \Nf{m}(B(a,a,\ReflConst(a))))\\
  &\quad \to \DeclVar{a_0,a_1}{\Open[z] \Tm{n}(z,A(z))} \DeclVar{p}{\Ne{n}(\IdConst(A,a_0,a_1))}\\
  &\quad \to \Ne{m}(B(a_0,a_1,\eta(p)))
\end{align*}
These changes simply reflect the change to the elimination principle of the identity type.

After having made this change, only one portion of Section~\ref{sec:model:mtt-cosmos} must change:
Lemma~\ref{lem:model:id} which shows that the gluing cosmos is closed under identity types. We must
show that $\prn{\Ty*{m},\Tm*{m}}$ is closed under crisp induction.

\begin{lem}
  $\prn{\Ty*{m},\Tm*{m}}$ supports crisp identity induction.
\end{lem}
\begin{proof}
  This argument is similar to Lemma \ref{lem:model:mod}, as the induction principle for modal types is
  always `crisp' in \MTT{}. We must implement the following constant.
  \begin{align*}
    &\IdElimConst*_\mu : \DeclVar{A}{\Ty*{n}}
    \,(B : \DeclVar{a_0,a_1}{\Tm*{n}(A)}\DeclVar{p}{\Tm*{n}(\IdConst*(A,a_0,a_1))} \to \Ty*{m})\\
    &\quad \to (b : \DeclVar{a}{\Tm*{n}(A)} \to \Tm*{m}(B(a,a,\ReflConst*(a))))\\
    &\quad \to \DeclVar{a_0,a_1}{\Tm*{n}(A)}\DeclVar{p}{\Tm*{n}(\IdConst(A,a_0,a_1))} \to\\
    &\quad \to \Ext{\Tm*{m}(B(a_0,a_1,p))}{z : \Syn}{\IdElimConst_\mu\prn{A,B,b,p}}
  \end{align*}

  Let us fix $A$, $B$, $b$, $a_0$, $a_1$, and $p$ with the types described above. Recalling the
  definition of $\IdConst*\prn{A,a_0,a_1}.\Pred$ from Lemma \ref{lem:model:id}, we can commute
  $\Modify[\mu]{-}$ past the dependent sum, closed modalities, equality types, and coproducts to
  decompose $p$ into a pair of the following:
  \begin{align*}
    &\DeclVar{\TmFld}{\Nf{n}(\IdConst(A,a_0,a_1))}
    \\
    &\PrfFld :
      \Closed \brk{
      \DelimMin{4}
      \begin{aligned}
        &(\Sum{e : \Modify[\mu]{\Ne{n}(\IdConst(A,a_0,a_1))}} \CNfInj \ZApp e = \MkBox{\TmFld})
        \\
        &+ (\MkBox{a_0} = \MkBox{a_1} \times \MkBox{\TmFld} = \MkBox{\NfRefl{a_0}})
      \end{aligned}
    }
  \end{align*}

  We then define $\IdElimConst*_\mu\prn{B,b,a_0,a_1,p}$ by analyzing $\PrfFld$:
  \[
    \begin{cases}
      \IdElimConst(z, B, b, a_0, a_1, p) & \PrfFld = \In{1}(z)
      \\
      \Reify{}
      \CIdRec(
        \lambda a_0,a_1,p.\ B(\Reflect{} a_0, \Reflect{} a_1, \Reflect{} p).\Code{},
        \lambda a.\ \Reify{} b(\Reflect{} a),
        e)
      & q = \In{2}(\In{1}(e, \_))
      \\
      b(a_0) & q = \In{2}(\In{2}(\_, \_)) \qedhere
    \end{cases}
  \]
\end{proof}

Having made this alteration, the remainder of Sections~\ref{sec:normalization-cosmos} and
\ref{sec:normalization} are unchanged. In particular, all the results of
Section~\ref{sec:normalization} continue to hold in the presence of crisp induction.

\section{Related work}
\label{sec:related}

We have built on top of a long line of research systematically structuring logical relations as gluing
models~\cite{mitchell:1993,altenkirch:1995,streicher:1998,fiore:2002,shulman:2015,altenkirch:2017,kaposi:gluing:2019,coquand:2019,sterling:2021,sterling:phd}.
In particular, Altenkirch et al.~\cite{altenkirch:1995} and Fiore~\cite{fiore:2002} recast NbE
into the construction of a gluing model in which types are triples
$\prn{A,\Reify{},\Reflect{}}$. Generalizing from this work to dependent type theory has proven a
considerable challenge~\cite{altenkirch:2016}. The final ingredient for Martin-L{\"o}f type theory
was provided by Coquand~\cite{coquand:2019}: a construction of a universe in this gluing model similar to
that of Shulman~\cite{shulman:2015}.

\paragraph{Gluing for modal type theory}
Gratzer et al.~\cite{gratzer:2019} gave a classical normalization-by-evaluation proof for a
Fitch-style type theory. The complexity of this proof, however, makes it intractable to extend to a
general modal type theory like \MTT{}. Unfortunately, extending gluing techniques to modal type
theories has proven challenging. In particular, Gratzer et al.~\cite{gratzer:2020} used gluing to
prove canonicity for \MTT{}, but they were forced to add an additional equality to \MTT{}
($\LockCx{\EmpCx} = \EmpCx$) to tame the construction of the gluing model. The challenge lies in
fitting the glued category of contexts into a CwF-style model of type theory; the natural definition
of glued types and terms fails to admit modalities. While there have been some attempts to
systematize the construction of glued CwFs~\cite{kaposi:gluing:2019}, they do not apply to \MTT{}.

Recently, Hu and Pientka~\cite{hu:2022} gave a proof of normalization for a simply-typed Fitch-style
type theory (Kripke-style in their parlance) with one modality. They give two separate proofs of
normalization; one through both an untyped PER model similar to Gratzer et al.~\cite{gratzer:2019}
and one using a gluing model. Their gluing proof is closely related to the argument above. For
instance, their theory of unified substitutions and modal transformations corresponds to a
specialization of \MTT{}'s substitution calculus to one modality and, accordingly, their category of
renamings offers a strict presentation of the category of renamings described above. Their proof,
however, is done using external constructions on the gluing category which may make it difficult to
scale to either multiple modalities or dependent types.

\paragraph{Synthetic Tait computability}
The introduction of representable map categories~\cite{uemura:2019} and
LCCCs~\cite{gratzer:lcccs:2020} for modeling the syntax of (non-modal) type theory offered an
alternative approach. Crucially, they show that syntax can be given a universal property among
structured categories with better behavior than CwFs. Sterling and
collaborators~\cite{sterling:modules:2021,sterling:2021,sterling:phd} have built on this idea and
introduced synthetic Tait computability to prove syntactic metatheorems via gluing together LCCCs
rather than CwFs. Unlike other approaches to gluing, STC generalizes well to a multimodal setting
and by extending STC to MSTC normalization for \MTT{} becomes tractable.

\paragraph{\MTT{} as a metalanguage}
In a parallel line of work, Bocquet et al.~\cite{bocquet:2021} have also used \MTT{} as a
metalanguage in the construction of models of type theory. They, however, do not work with a modal
object type theory and instead use \MTT{} to internalize a functor $F$ rather than working
internally to $\GL{F}$. As a result, while both proofs use \MTT{} modalities, the modalities used by
\opcit{} are encoded in our proof by fibered lex monads ($\Open$, $\Closed$) which prove easier to
manipulate.

\section{Conclusions and future work}
\label{sec:conclusions}

We prove normalization for \MTT{} (Theorem \ref{thm:normalization:normalization}) and thereby reduce
the decidability of conversion and type checking to the decidability of equality of the underlying
mode theory (Corollaries \ref{cor:normalization:dec} and \ref{cor:normalization:type-checking}). In
addition, we deduce a number of corollaries from normalization itself, including the injectivity of
type constructors and canonicity (Corollaries \ref{cor:normalization:pi-inj} and
\ref{cor:normalization:canonicity}).

By working constructively, we have obtained an effective procedure for normalization. This, along
with our results on type checking, open the door to a theoretically-sound implementation of \MTT{}
generic in the mode theory. In the future, we intend to develop a bidirectional syntax for \MTT{}
and implement it. Stassen et al.~\cite{stassen:2022} have made promising initial steps in this
direction for \emph{poset-enriched} mode theories.

\section*{Acknowledgments}
I am thankful for discussions with Carlo Angiuli, Martin Bidlingmaier, Lars Birkedal, Thierry
Coquand, Alex Kavvos, Christian Sattler, and Jonathan Sterling. I am also grateful to the careful
reading and comments provided by the reviewers of this paper.
The author was supported in part by a Villum Investigator grant (no. 25804), Center for Basic
Research in Program Verification (CPV), from the VILLUM Foundation.

\bibliographystyle{alphaurl}
\bibliography{refs,temp-refs}

\end{document}